%% file: draft.tex
\pdfoutput=1
\documentclass[sigconf, dvipsnames]{acmart}
\usepackage{hyperref}
\usepackage{enumerate}

\usepackage{algorithm}
\usepackage[noend]{algpseudocode}

\usepackage{caption}
\usepackage[subrefformat=parens,labelformat=parens]{subcaption}
\DeclareCaptionSubType * [alph]{table}
\graphicspath{ {images1/}}

\def\R{\mathbb{R}}
\def\F{\mathcal{F}}

\definecolor{shadecolor}{gray}{0.9}

\newcommand{\card}[1]{\ensuremath{|#1|}}

\def\etal{\text{et al.}}

\def\part{\texttt{part}}

\newcommand{\squishlisttwo}{
   \begin{list}{$\bullet$}
       { \setlength{\itemsep}{0pt}      \setlength{\parsep}{3pt}
       \setlength{\topsep}{3pt}       \setlength{\partopsep}{0pt}
      \setlength{\leftmargin}{1.5em} \setlength{\labelwidth}{1em}
       \setlength{\labelsep}{0.5em} } }

        {\hspace*{\fill}$\Box$\par}

\newcounter{Lcount}
\newcommand{\squishlist}{
\begin{list}{\arabic{Lcount}. }
{ \usecounter{Lcount}
\setlength{\itemsep}{0pt}
\setlength{\parsep}{0pt}
\setlength{\topsep}{0pt}
\setlength{\partopsep}{0pt}
\setlength{\leftmargin}{2em}
\setlength{\labelwidth}{1.5em}
\setlength{\labelsep}{0.5em} } }

\newcommand{\squishend}{
\end{list} }

\setcopyright{none}

\begin{document}
\pagestyle{plain}

\newtheorem{remark}[theorem]{Remark}
\newtheorem{claim}{Claim}
\newtheorem{property}[theorem]{Property}

\title{A Signaling Game Approach to Data Querying and Interaction}

\author{Ben McCamish}
\affiliation{
	\institution{Oregon State University}
}
\email{mccamisb@oregonstate.edu}

\author{Vahid Ghadakchi}
\affiliation{
	\institution{Oregon State University}
}
\email{ghadakcv@oregonstate.edu}

\author{Arash Termehchy}
\affiliation{
	\institution{Oregon State University}
}
\email{termehca@oregonstate.edu}

\author{Behrouz Touri}
\affiliation{
	\institution{University of California San Diego}
}
\email{btouri@eng.ucsd.edu}

\renewcommand{\shortauthors}{B. McCamish et al.}

\begin{abstract}
As most users do {\it not} precisely know the structure and/or the content of databases, their queries do {\it not} exactly reflect their information needs. The database management systems (DBMS) may interact with users and use their feedback on the returned results to learn the information needs behind their queries. Current query interfaces assume that users do {\it not} learn and modify the way way they express their information needs in form of queries during their interaction with the DBMS. Using a real-world interaction workload, we show that users learn and modify how to express their information needs during their interactions with the DBMS and their learning is accurately modeled by a well-known reinforcement learning mechanism. As current data interaction systems assume that users do {\it not} modify their strategies, they cannot discover the information needs behind users' queries effectively. We model the interaction between users and DBMS as a game with identical interest between two rational agents whose goal is to establish a common language for representing information needs in form of queries. We propose a reinforcement learning method that learns and answers the information needs behind queries and adapts to the changes in users' strategies and prove that it improves the effectiveness of answering queries stochastically speaking. We propose two efficient implementation of this method  over large relational databases. Our extensive empirical studies over real-world query workloads indicate that our algorithms are efficient and effective. 
\end{abstract}

\maketitle
\input{1.Introduction.tex}
\input{2.Framework.tex}
\input{4.UserStudy.tex}
\input{5.Theory.tex}

\input{6.Relational.tex}
\input{7.Comparison.tex}

\input{3.Equilibria.tex}

\input{8.RelatedWork.tex}
\input{9.Conclusion.tex}
\bibliographystyle{ACM-Reference-Format}
\bibliography{ref}

\end{document}

%% file: 1.Introduction.tex
\section{Introduction}
\label{sec:introduction}
Most users do not know the structure and content of databases and concepts such as schema or formal query languages sufficiently well to express their information needs precisely in the form of queries~\cite{Usable:Jagadish,Chen:2009:KSS,Idreos:SIGMOD:2015}. They may convey their intents in easy-to-use but inherently ambiguous forms, such as keyword queries, which are open to numerous interpretations. Thus, it is very challenging for a database management system (DBMS) to understand and satisfy the intents behind these queries. The fundamental challenge in the interaction of these users and DBMS is that the users and DBMS represent intents in different forms.

Many such users may explore a database to find answers for various intents over a rather long period of time. For these users, database querying is an inherently interactive and continuing process. As both the user and DBMS have the same goal of the user receiving her desired information, the user and DBMS would like to gradually improve their understandings of each other and reach a {\it common language of representing intents} over the course of various queries and interactions. The user may learn more about the structure and content of the database and how to express intents as she submits queries and observes the returned results. Also, the DBMS may learn more about how the user expresses her intents by leveraging user feedback on the returned results. The user feedback may include clicking on the relevant answers \cite{Yue:2012:KDB:2240304.2240501}, the amount of time the user spends on reading the results \cite{Granka:2004:EAU:1008992.1009079}, user's eye movements \cite{Huang:2012:USU:2207676.2208591}, or the signals sent in touch-based devises \cite{DBLP:conf/icde/LiarouI14}. Ideally, the user and DBMS should establish as quickly as possible this common representation of intents in which the DBMS accurately understands all or most user's queries. 

Researchers have developed systems that leverage user feedback to help the DBMS understand the intent behind ill-specified and vague queries more precisely \cite{Chaudhuri:2006:PIR:1166074.1166085,Chatzopoulou:2009:QRI:1561638.1561642}. These systems, however, generally assume that a user does {\it not} modify her method of expressing intents throughout her interaction with the DBMS. For example, they maintain that the user picks queries to express an intent according to a fixed probability distribution. It is known that the learning methods that are useful in a static setting do not deliver desired outcomes in a setting where all agents may modify their strategies \cite{Grotov:2016:OLR:2911451.2914798,Daskalakis:2010:LAN:1929237.1929248}. Hence, one may not be able to use current techniques to help the DBMS understand the users' information need in a rather long-term interaction. 

To the best of our knowledge, the impact of user learning on database interaction has been generally ignored. In this paper, we propose a novel framework that formalizes the interaction between the user and the DBMS as a game with identical interest between  two active and potentially rational agents: the user and DBMS. The common goal of the user and DBMS is to reach a mutual understanding on expressing information needs in the form of keyword queries. In each interaction, the user and DBMS receive certain payoff according to how much the returned results are relevant to the intent behind the submitted query. The user receives her payoff by consuming the relevant information and the DBMS becomes aware of its payoff by observing the user's feedback on the returned results. We believe that such a game-theoretic framework naturally models the long-term interaction between the user and DBMS. We explore the user learning mechanisms and propose algorithms for DBMS to improve its understanding of intents behind the user queries effectively and efficiently over large databases. In particular, we make the following contributions:

\squishlisttwo
    \item We model the long term interaction between the user and DBMS using keyword queries as a particular type of game called a signaling game~\cite{Cho:QuarterlyJournalEconomics:1987} in Section~\ref{sec:framework}. 
   
    \item Using extensive empirical studies over a real-world interaction log, we show that users modify the way they express their information need over their course of interactions in Section~\ref{sec:results}. We also show that this adaptation is accurately modeled by a well-known reinforcement learning algorithm \cite{roth1995learning} in experimental game-theory. 
 
 	\item Current systems generally assume that a user does {\it not} learn and/or modify her method of expressing intents throughout her interaction with the DBMS. However, it is known that the learning methods that are useful in static settings do not deliver desired outcomes in the dynamic ones \cite{auer2002nonstochastic}. We propose a method of answering user queries in a natural and interactive setting in Section~\ref{sec:learning} and prove that it improves the effectiveness of answering queries stochastically speaking, and converges almost surely. We show that our results hold for both the cases where the user adapts her strategy using an appropriate learning algorithm and the case where she follows a fixed strategy.

	\item We describe our data interaction system that provides an efficient implementation of our reinforcement learning method on large relational databases in Section~\ref{sec:relational}. In particular, we first propose an algorithm that implements our learning method called {\it Reservoir}. Then, using certain mild assumptions and the ideas of sampling over relational operators, we propose another algorithm called {\it Poisson-Olken} that implements our reinforcement learning scheme and considerably improves the efficiency of {\it Reservoir}.

	\item We report the results of our extensive empirical studies on measuring the effectiveness  of our reinforcement learning method and  the efficiency of our algorithms using real-world and large interaction workloads, queries, and databases in Section~\ref{sec:comparison}. Our results indicate that our proposed reinforcement learning method is more effective than the start-of-the-art algorithm for long-term interactions. They also show that {\it Poisson-Olken} can process queries over large databases faster than the {\it Reservoir} algorithm.
	
	\item Finally, we formally analyze the eventual stable states and equilibria of the data interaction game in Section~\ref{sec:equilib}. We show that the game has generally several types of equilibria some of which are {\it not} desirable.

\end{list}

%% file: 2.Framework.tex
\section{A Game-theoretic Framework}
\label{sec:framework}
Users and DBMSs typically achieve a common understanding {\it gradually} and using a {\it querying/feedback} paradigm. After submitting each query, the user may revise her strategy of expressing intents based on the returned result. If the returned answers satisfy her intent to a large extent, she may keep using the same query to articulate her intent. Otherwise, she may revise her strategy and choose another query to express her intent in the hope that the new query will provide her with more relevant answers. We will describe this behavior of users in Section~\ref{sec:results} in more details. The user may also inform the database system about the degree by which the returned answers satisfy the intent behind the query using explicit or implicit feedback, e.g., click-through information \cite{Granka:2004:EAU:1008992.1009079}. The DBMS may update its interpretation of the query according to the user's feedback. 

Intuitively, one may model this interaction as a game between two agents with identical interests in which the agents communicate via sharing queries, results, and feedback on the results. In each interaction, both agents will receive some reward according to the degree by which the returned result for a query matches its intent. The user receives her rewards in the form of answers relevant to her intent and the DBMS receives its reward through getting positive feedback on the returned results. The final goal of both agents is to maximize the amount of reward they receive during the course of their interaction. Next, we describe the components and structure of this interaction game for relational databases. 

{\bf Basic Definitions:} We fix two disjoint (countably) infinite sets of attributes and relation symbols. Every relation symbol $R$ is associated with a set of attribute symbols denoted as $sort(R)$. Let $dom$ be a countably infinite set of constants, e.g., strings. An instance $I_R$ of relation symbol $R$ with $n = \card{sort(R)}$ is a (finite) subset of $dom^n$. A {\it schema} $S$ is a set of relation symbols. A database (instance) of $S$ is a mapping over $S$ that associates with each relation symbol $R$ in $S$ an instance of $I_R$. In this paper, we assume that $dom$ is a set of strings.

\subsection{Intent} 
An {\it intent} represents an information need sought after by the user. Current keyword query interfaces over relational databases generally assume that each intent is a query in a sufficiently expressive query language in the domain of interest, e.g., Select-Project-Join subset of SQL \cite{Chen:2009:KSS,Usable:Jagadish}. Our framework and results are orthogonal to the language that precisely describes the users' intents. Table~\ref{example:framework:table:instance} illustrates a database with schema \textit{Univ(Name, Abbreviation, State, Type, Ranking)} that contains information about university rankings. A user may want to find the information about university {\it MSU} in Michigan, which is precisely represented by the intent $e_2$ in Table~\ref{example:framework:table:intents}, which using the Datalog syntax \cite{AliceBook} is: $ans(z)$ $\leftarrow$ $Univ(x,`MSU\textrm', `MI\textrm', y, z)$.

\subsection{Query}
Users' articulations of their intents are {\it queries}. Many users do not know the formal query language, e.g., SQL, that precisely describes their intents. Thus, they may prefer to articulate their intents in languages that are easy-to-use, relatively less complex, and ambiguous such as keyword query language~\cite{Usable:Jagadish,Chen:2009:KSS}. In the proposed game-theoretic frameworks for database interaction, we assume that the user expresses her intents as keyword queries. More formally, we fix a countably infinite set of terms, i.e., keywords, $T$. A {\it keyword query} (query for short) is a nonempty (finite) set of terms in $T$. Consider the database instance in Table~\ref{example:framework:table:instance}. Table~\ref{example:framework:table:intentsandqueries} depicts a set of intents and queries over this database. Suppose the user wants to find the information about Michigan State University in Michigan, i.e. the intent $e_2$. Because the user does not know any formal database query language and may not be sufficiently familiar with the content of the data, she may express intent $e_2$ using $q_2:$ {\it `MSU'}.

Some users may know a formal database query language that is sufficiently expressive to represent their intents. Nevertheless, because they may not know precisely the content and schema of the database, their submitted queries may not always be the same as their intents \cite{Chaudhuri:2006:PIR:1166074.1166085,Khoussainova:2010:SCA:1880172.1880175}. For example, a user may know how to write a SQL query. But, since she may not know the state abbreviation {\it MI}, she may articulate intent $e_2$ as $ans(t)$ $\leftarrow$ $Univ(x,`MSU\textrm', y, z, t)$, which is different from $e_2$. We plan to extend our framework for these scenarios in future work. But, in this paper, we assume that users articulate their intents as keyword queries.

\subsection{User Strategy} 
The user strategy indicates the likelihood by which the user submits query $q$ given that her intent is $e$. In practice, a user has finitely many intents and submits finitely many queries in a finite period of time. Hence, we assume that the sets of the user's intents and queries are finite. We index each user's intent and query by $1 \leq i \leq m$ and $1 \leq j \leq n$, respectively. A user strategy, denoted as $U$, is a $m\times n$ row-stochastic matrix from her intents to her queries. The matrix on the top of Table~\ref{example:framework:table:badstrategy:nash} depicts a user strategy using intents and queries in Table~\ref{example:framework:table:intentsandqueries}. According to this strategy, the user submits query $q_2$ to express intents $e_1$, $e_2$, and $e_3$.

\begin{table}[ht!]
        \centering
        \small
    	\caption{A database instance of relation Univ}
        \begin{tabular}{l l l l l}
            \hline
            \hline
            Name & Abbreviation & State & Type & Rank\\
            \hline
            Missouri State University & MSU & MO & public & 20\\
            Mississippi State University & MSU & MS & public & 22\\
            Murray State University & MSU & KY & public & 14\\
            Michigan State University & MSU & MI & public & 18\\
            \hline
        \end{tabular}
    \label{example:framework:table:instance}
\end{table}

\begin{table}[htbp!]
	\caption{Intents and Queries}
    \centering
    \small
    \begin{subtable}{1\linewidth}
    	\centering
 		\caption{Intents}
        \begin{tabular}{l l}
            \hline
            \hline
            Intent\# & \multicolumn{1}{c}{Intent} \\
            \hline
 	        $e_1$ & $ans(z) \leftarrow Univ(x,`MSU\textrm', `MS\textrm', y, z)$\\
            $e_2$ & $ans(z) \leftarrow Univ(x,`MSU\textrm', `MI\textrm', y, z)$\\
  	        $e_3$ & $ans(z) \leftarrow Univ(x,`MSU\textrm', `MO\textrm', y, z)$\\
            \hline
        \end{tabular}
        \label{example:framework:table:intents}
    \end{subtable}%
    
    \begin{subtable}{1\linewidth}
    	\centering
        \caption{Queries}
        \begin{tabular}{l l}
            \hline
            \hline
            Query\# & \multicolumn{1}{c}{Query} \\
            \hline
            $q_1$ & `MSU MI'\\
            $q_2$ & `MSU'\\
            \hline
        \end{tabular}
        \label{example:framework:table:queries}
    \end{subtable}%
    \label{example:framework:table:intentsandqueries}
\end{table}
\begin{table}[ht!]
	\caption{Two strategy profiles over the intents and queries in Table~\ref{example:framework:table:intentsandqueries}. User and DBMS strategies at the top and bottom, respectively.}
	\centering
    \begin{subtable}{.45\linewidth}
    	\caption{A strategy profile}
        \begin{tabular}{c|c|c|}
            \cline{2-3}
             & \multicolumn{1}{c|}{$q_1$} &$q_2$\\
            \hline
            \multicolumn{1}{|c|}{$e_1$} & 0 & 1\\
            \hline
            \multicolumn{1}{|c|}{$e_2$} & 0 & 1\\
            \hline
            \multicolumn{1}{|c|}{$e_3$} & 0 & 1\\
            \hline
        \end{tabular}
        \begin{tabular}{c|c|c|c|}
            \cline{2-4}
             & \multicolumn{1}{c|}{$e_1$} & $e_2$ & $e_3$ \\
            \hline
            \multicolumn{1}{|c|}{$q_1$} & 0 & 1 & 0\\
            \hline
            \multicolumn{1}{|c|}{$q_2$} & 0 & 1 & 0\\
            \hline
        \end{tabular}
       	\label{example:framework:table:badstrategy:nash}
    \end{subtable}	
    \begin{subtable}{.45\linewidth}
    	\caption{Another strategy profile}
        \begin{tabular}{c|c|c|}
            \cline{2-3}
             & \multicolumn{1}{c|}{$q_1$} &$q_2$\\
            \hline
            \multicolumn{1}{|c|}{$e_1$} & 0 & 1\\
            \hline
            \multicolumn{1}{|c|}{$e_2$} & 1 & 0\\
            \hline
            \multicolumn{1}{|c|}{$e_3$} & 0 & 1\\
            \hline
        \end{tabular}
        \begin{tabular}{c|c|c|c|}
            \cline{2-4}
             & \multicolumn{1}{c|}{$e_1$} & $e_2$ & $e_3$ \\
            \hline
            \multicolumn{1}{|c|}{$q_1$} & 0 & 1 & 0\\
            \hline
            \multicolumn{1}{|c|}{$q_2$} & 0.5 & 0 & 0.5\\
            \hline
        \end{tabular}
        \label{example:framework:table:beststrategy:nash}
    \end{subtable}
    \label{example:framework:table:strategies}
\end{table}

\subsection{DBMS Strategy} 
\label{sec:DBStrategy}
The DBMS interprets queries to find the intents behind them. It usually interprets queries by mapping them to a subset of SQL \cite{Chen:2009:KSS,IRStyle,spark}. Since the final goal of users is to see the result of applying the interpretation(s) on the underlying database, the DBMS runs its interpretation(s) over the database and returns its results. Moreover, since the user may {\it not} know SQL, suggesting possible SQL queries may not be useful. A DBMS may {\it not} exactly know the language that can express all users' intents. Current usable query interfaces, including keyword query systems, select a query language for the interpreted intents that is sufficiently complex to express many users' intents and is simple enough so that the interpretation and running its outcome(s) are done efficiently \cite{Chen:2009:KSS}. As an example consider current keyword query interfaces over relational databases~\cite{Chen:2009:KSS}. Given constant $v$ in database $I$ and keyword $w$ in keyword query $q$, let $match(v,w)$ be a function that is true if $w$ appears in $v$ and false otherwise. A majority of keyword query interfaces interpret keyword queries as Select-Project-Join queries that have below certain number of joins and whose {\it where} clauses contain only conjunctions of $match$ functions~\cite{IRStyle,spark}. Using a larger subset of SQL, e.g. the ones with more joins, makes it inefficient to perform the interpretation and run its outcomes. Given schema $S$, the {\it interpretation language} of the DBMS, denoted as $L$, is a subset of SQL over $S$. We precisely define $L$ for our implementation of DBMS strategy in Section~\ref{sec:relational}. To interpret a keyword query, the DBMS searches $L$ for the SQL queries that represent the intent behind the query as accurately as possible. 

Because users may be overwhelmed by the results of many interpretations, keyword query interfaces use a deterministic real-valued scoring function to rank their interpretations and deliver only the results of top-$k$ ones to the user \cite{Chen:2009:KSS}. It is known that such a deterministic approach may significantly limit the accuracy of interpreting queries in long-term interactions in which the information system utilizes user's feedback~\cite{hofmann2013balancing,vorobev2015gathering,auer2002finite}. Because the DBMS shows only the result of interpretation(s) with the highest score(s) to the user, it receives feedback only on a small set  of interpretations. Thus, its learning remains largely biased toward the initial set of highly ranked interpretations. For example, it may never learn that the intent behind a query is satisfied by an interpretation with a relatively low score according to the current scoring function. 

To better leverage users feedback during the interaction, the DBMS must show the results of and get feedback on a sufficiently diverse set of interpretations~\cite{hofmann2013balancing,vorobev2015gathering,auer2002finite}. Of course, the DBMS should ensure that this set of interpretations are relatively relevant to the query, otherwise the user may become discouraged and give up querying. This dilemma is called the {\it exploitation versus exploration} trade-off. A DBMS that only {\it exploits}, returns top-ranked interpretations according to its scoring function. Hence, the DBMS may adopt a {\it stochastic strategy} to both exploit and explore: it randomly selects and shows the results of intents such that the ones with higher scores are chosen with larger probabilities~\cite{hofmann2013balancing,vorobev2015gathering,auer2002finite}. In this approach, users are mostly shown results of interpretations that are relevant to their intents according to the current knowledge of the DBMS and provide feedback on a relatively diverse set of interpretations. More formally, given $Q$ is a set of all keyword queries, the DBMS strategy $D$ is a stochastic mapping from $Q$ to $L$. To the best of our knowledge, to search $L$ efficiently, current keyword query interfaces limit their search per query to a finite subset of $L$ \cite{Chen:2009:KSS,IRStyle,spark}. In this paper, we follow a similar approach and assume that $D$ maps each query to only a finite subset of $L$. The matrix on the bottom of Table~\ref{example:framework:table:badstrategy:nash} depicts a DBMS strategy for the intents and queries in Table~\ref{example:framework:table:intentsandqueries}. Based on this strategy, the DBMS uses a exploitative strategy and always interprets query $q_2$ as $e_2$. The matrix on the bottom of Table~\ref{example:framework:table:beststrategy:nash} depicts another DBMS strategy for the same set of intents and queries. In this example, DBMS uses a randomized strategy and does both exploitation and exploration. For instance, it explores $e_1$ and $e_2$ to answer $q_2$ with equal probabilities, but it always returns $e_2$ in the response to $q_1$.

\vspace{-4pt}
\subsection{Interaction \& Adaptation}
The data interaction game is a repeated game with identical interest between two players, the user and the DBMS. At each round of the game, i.e., a single interaction, the user selects an intent according to the prior probability distribution $\pi$. She then picks the query $q$ according to her strategy and submits it to the DBMS. The DBMS observes $q$ and interprets $q$ based on its strategy, and returns the results of the interpretation(s) on the underlying database to the user. The user provides some feedback on the returned tuples and informs the DBMS how relevant the tuples are to her intent. In this paper, we assume that the user informs the DBMS if some tuples satisfy the intent via some signal, e.g., selecting the tuple, in some interactions. The feedback signals may be noisy, e.g., a user may click on a tuple by mistake. Researchers have proposed models to accurately detect the informative signals \cite{hofmann2013balancing}. Dealing with the issue of noisy signals is out of the scope of this paper.

The goal of both the user and the DBMS is to have as many satisfying tuples as possible in the returned tuples. Hence, both the user and the DBMS receive some payoff, i.e., reward, according to the degree by which the returned tuples match the intent. This payoff is measured based on the user feedback and using standard effectiveness metrics \cite{IRBook}. One example of such metrics is {\it precision at $k$}, $p@k$, which is the fraction of relevant tuples in the top-$k$ returned tuples. At the end of each round, both the user and the DBMS receive a payoff equal to the value of the selected effectiveness metric for the returned result. We denote the payoff received by the players at each round of the game, i.e., a single interaction, for returning interpretation $e_{\ell}$ for intent $e_i$ as $r(e_i,e_{\ell})$. This payoff is computed using the user's feedback on the result of interpretation $e_{\ell}$ over the underlying database.

Next, we compute the expected payoff of the players. Since DBMS strategy $D$ maps each query to a finite set of interpretations, and the set of submitted queries by a user, or a population of users, is finite, the set of interpretations for all queries submitted by a user, denoted as $L^s$, is finite. Hence, we show the DBMS strategy for a user as an $n\times o$ row-stochastic matrix from the set of the user's queries to the set of interpretations $L^s$. We index each interpretation in $L^s$ by $1 \leq \ell \leq o$. Each pair of the user and the DBMS strategy, ($U$,$D$), is a {\it strategy profile}. The expected payoff for both players with strategy profile ($U$,$D$) is as follows.
\begin{align}
\label{eqn:payoff}
u_{r}(U,D) =\sum_{i=1}^m\pi_i\sum_{j=1}^nU_{ij}\sum_{\ell=1}^o D_{j\ell}\ r(e_i,e_{\ell}), 
\end{align}
The expected payoff reflects the degree by which the user and DBMS have reached a common language for communication. This value is high for the case in which the user knows which queries to pick to articulate her intents and the DBMS returns the results that satisfy the intents behind the user's queries. Hence, this function reflects the success of the communication and interaction. For example, given that all intents have equal prior probabilities, intuitively, the strategy profile in Table~\ref{example:framework:table:beststrategy:nash} shows a larger degree of mutual understanding between the players than the one in Table~\ref{example:framework:table:badstrategy:nash}. This is reflected in their values of expected payoff as the expected payoffs of the former and latter are $\frac{2}{3}$ and $\frac{1}{3}$, respectively. We note that the DBMS may {\it not} know the set of users' queries beforehand and does {\it not} compute the expected payoff directly. Instead, it uses query answering algorithms that leverage user feedback, such that the expected payoff improves over the course of several interactions as we will show in Section~\ref{sec:learning}. 

None of the players know the other player's strategy during the interaction. Given the information available to each player, it may modify its strategy at the end of each round (interaction). For example, the DBMS may reduce the probability of returning certain interpretations that has not received any positive feedback from the user in the previous rounds of the game. Let the user and DBMS strategy at round $t \in \mathbb{N}$ of the game be $U(t)$ and $D(t)$, respectively. In round $t \in \mathbb{N}$ of the game, the user and DBMS have access to the information about their past interactions. The user has access to her sequence of intents, queries, and results, the DBMS knows the sequence of queries and results, and both players have access to the sequence of payoffs ({\it not} expected payoffs) up to round $t-1$. It depends on the degree of rationality and abilities of the user and the DBMS  how to leverage these pieces of information to improve the expected payoff of the game. For example, it may {\it not} be reasonable to assume that the user adopts a mechanism that requires instant access to the detailed information about her past interactions as it is not clear whether users can memorize this information for a long-term interaction. A {\it data interaction game} is represented as tuple $(U(t), D(t), \pi, (e^u(t-1)), (q(t-1)), (e^d(t-1)), (r(t-1)))$ in which $U(t)$ and $D(t)$ are respectively the strategies of the user and DBMS at round $t$, $\pi$ is the prior probability of intents in $U$, $(e^u(t-1))$ is the sequence of intents, $(q(t-1))$ is the sequence of queries, $(e^d(t-1))$ is the sequence of interpretations, and  $(r(t-1)))$ is the sequence of payoffs up to time $t$. Table~\ref{framework:notation} contains the notation and concept definitions introduced in this section for future reference.

\begin{table}
    \centering
    \caption{Summary of the notations used in the model.}
    \vspace{-10pt}
    \label{framework:notation}
    \begin{tabular}{|l|p{6.7cm}|}
    \hline
    Notation & Definition\\
    \hline
    $e_i$ & A user's intent\\
    \hline
    $q_j$ & A query submitted by the user \\
    \hline
    $\pi_i$ & The prior probability that the user queries for $e_i$\\
    \hline
    $r(e_i,e_\ell)$ & The reward when the user looks for $e_i$ and the DBMS returns $e_\ell$\\
    \hline
    $U$ & The user strategy\\
    \hline
    $U_{ij}$ & The probability that user submits $q_j$ for intent $e_i$\\
    \hline
    $D$ & The DBMS strategy\\
    \hline
    $D_{j\ell}$ & The probability that DBMS intent $e_{\ell}$ for query $q_j$\\
    \hline
    $(U,D)$ & A strategy profile\\
    \hline
    $u_r(U,D)$ & The expected payoff of the strategy profile $(U,D)$ computed using reward metric $r$ based to Equation~1\\
    \hline
    \end{tabular}
    \vspace{-15pt}
\end{table}

%% file: 4.UserStudy.tex
\section{User Learning Mechanism}
\label{sec:results}
It is well established that humans show reinforcement behavior in learning \cite{4176,ICMLRL}. Many lab studies with human subjects conclude that one can model human learning using reinforcement learning models \cite{4176,ICMLRL}. The exact reinforcement learning method used by a person, however, may vary based on her capabilities and the task at hand. We have performed an empirical study of a real-world interaction log to find the reinforcement learning method(s) that best explain the mechanism by which users adapt their strategies during interaction with a DBMS. 

\subsection{Reinforcement Learning Methods}
To provide a comprehensive comparison, we evaluate six reinforcement learning methods used to model human learning in experimental game theory and/or Human Computer Interaction (HCI)~\cite{roth1995learning,cen2013reinforcement}. These methods mainly vary based on 1) the degree by which the user considers past interactions when computing future strategies, 2) how they update the user strategy, and 3) the rate by which they update the user strategy. {\it Win-Keep/Lose-Randomize} keeps a query with non-zero reward in past interactions for an intent. If such a query does not exist, it picks a query randomly. {\it Latest-Reward} reinforces the probability of using a query to express an intent based on the most recent reward of the query to convey the intent. {\it Bush and Mosteller's} and {\it Cross's} models increases (decreases) the probability of using a query based its past success (failures) of expressing an intent. A query is successful if it delivers a reward more than a given threshold, e.g., zero. {\it Roth and Erev's} model uses the aggregated reward from past interactions to compute the probability by which a query is used. {\it Roth and Erev's modified} model is similar to Roth and Erev's model, with an additional parameter that determines to what extent the user {\it forgets} the reward received for a query in past interactions.

\subsubsection{Win-Keep/Lose-Randomize}
This method uses only the most recent interaction for an intent to determine the queries used to express the intent in the future \cite{Barrett}. Assume that the user conveys an intent~$e$ by a query~$q$. If the reward of using~$q$ is above a specified threshold~$\tau$ the user will use~$q$ to express~$e$ in the future. Otherwise, the user randomly picks another query uniformly at random to express $e$. 

\subsubsection{Bush and Mosteller's Model:}
Bush and Mosteller's model increases the probability that a user will choose a given query to express an intent by an amount proportional to the reward of using that query and the current probability of using this query for the intent~\cite{bush1953stochastic}. It also decreases the probabilities of queries not used in a successful interaction. If a user receives reward $r$ for using $q(t)$ at time $t$ to express intent $e_i$, the model updates the probabilities of using queries in the user strategy as follows.

\begin{equation}
  \label{eq:BM:jq}
  U_{ij}(t+1) =  
  \begin{cases}
      U_{ij}(t) + \alpha^{BM} \cdot (1-U_{ij}(t)) & q_j=q(t) \land r \geq 0\\
      U_{ij}(t) - \beta^{BM} \cdot U_{ij}(t) & q_j=q(t) \land r < 0\\
  \end{cases}
\end{equation}

\begin{equation}
  \label{eq:BM:jnotq}
  U_{ij}(t+1) =  
  \begin{cases}
      U_{ij}(t) - \alpha^{BM} \cdot U_{ij}(t) & q_j\neq q(t) \land r \geq 0\\
      U_{ij}(t) + \beta^{BM} \cdot (1-U_{ij}(t) & q_j\neq q(t)\land r < 0
  \end{cases}
\end{equation}

$\alpha^{BM} \in [0,1]$ and $\beta^{BM} \in [0,1]$ are parameters of the model. If query~$q_j$ is equal to  $q(t)$ then Equation~\ref{eq:BM:jq} is used. For all other queries $q_j$ for the intent $e_i$ at time~$t$, Equation~\ref{eq:BM:jnotq} is used. The probabilities of using queries for intents other than $e_i$ remains unchanged. Since effectiveness metrics in interaction are always greater than zero, $\beta^{BM}$ is never used in our experiments.

\subsubsection{Cross's Model:}
Cross's model modifies the user's strategy similar to Bush and Mosteller's model~\cite{cross1973stochastic}, but uses the amount of the received reward to update the user strategy. Given a user receives reward $r$ for using $q(t)$ at time $t$ to express intent $e_i$, we have:

\begin{equation}
  U_{ij}(t+1) =  
  \begin{cases}
      U_{ij}(t) + R(r) \cdot (1-U_{ij}(t)) & q_j=q(t)\\
      U_{ij}(t) - R(r) \cdot U_{ij}(t)   & q_j\neq q(t)
  \end{cases}
\end{equation}
\begin{equation}
  R(r) = \alpha^{C} \cdot r + \beta^{C}
\end{equation}

Parameters $\alpha^{C} \in [0,1]$ and $\beta^{C} \in [0,1]$ are used to compute the adjusted reward $R(r)$ based on the value of actual reward $r$. The parameter $\beta^{C}$ is a static increment of the adjusted reward. Similar to Bush and Mosteller's model, the aforementioned formulas are used to update the probabilities of using queries for the intent $e_i$ in the current interaction. Other entries in the user's strategy are remained unchanged.

\subsubsection{Roth and Erev's Model:}
Roth and Erev's model reinforces the probabilities directly from the reward value $r$ that is received when the user uses query $q(t)$ \cite{roth1995learning}. Its most important difference with other models is that it explicitly accumulates all the rewards gained by using a query to express an intent. $S_{ij}(t)$ in matrix $S(t)$ maintains the accumulated reward of using query $q_j$ to express intent $e_i$ over the course of interaction up to round (time) $t$. 

\begin{equation}
  S_{ij}(t+1) =  
  \begin{cases}
      S_{ij}(t) + r & q_j=q(t)\\
      S_{ij}(t) & q_j\neq q(t)
  \end{cases}
\end{equation}

\begin{equation}
  U_{ij}(t+1) = \frac{S_{ij}(t+1)}{\sum\limits_{j'}^n S_{ij'}(t+1)}
\end{equation}

Roth and Erev's model increases the probability of using a query to express an intent based on the accumulated rewards of using that query over the long-term interaction of the user. 
Each query {\it not} used in a successful interaction will be implicitly penalized as when the probability of a query increases, all others will decrease to keep $U$ row-stochastic.  

\subsubsection{Roth and Erev's Modified Model:}
Roth and Erev's modified model is similar to the original Roth and Erev's model, but it has an additional parameter that determines to what extent the user takes in to account the outcomes of her past interactions with the system~\cite{erev1995need}. It is reasonable to assume that the user may forget the results of her much earlier interactions with the system. User's memory is imperfect which means that over time the strategy may change merely due to the forgetful nature of the user. This is accounted for by the \textit{forget} parameter $\sigma \in [0,1]$. Matrix $S(t)$ has the same role it has for the Roth and Erev's model.

\begin{equation}
  S_{ij}(t+1) = (1-\sigma) \cdot S_{ij}(t) + E(j, R(r))
\end{equation}

\begin{equation}
  E(j, R(r)) =  
  \begin{cases}
      R(r) \cdot (1-\epsilon) & q_j=q(t)\\
      R(r) \cdot (\epsilon) & q_j \neq q(t)
  \end{cases}
\end{equation}

\begin{equation}
  R(r) = r - r_{min}
\end{equation}

\begin{equation}
  U_{ij}(t+1) = \frac{S_{ij}(t+1)}{\sum\limits_{j'}^n S_{ij'}(t+1)}
\end{equation}

In the aforementioned formulas, $\epsilon \in [0,1]$ is a parameter that weights the reward that the user receives, $n$ is the maximum number of possible queries for a given intent $e_i$, and $r_{min}$ is the minimum expected reward that the user wants to receive. The intuition behind this parameter is that the user often assumes some minimum amount of reward is guaranteed when she queries the database. The model uses this minimum amount to discount the received reward. We set $r_{min}$ to 0 in our analysis, representing that there is no expected reward in an interaction. 

\subsection{Empirical Analysis}
\label{sec:empAnal}
\subsubsection{Interaction Logs}
We use an anonymized Yahoo! interaction log for our empirical study, which consists of queries submitted to a Yahoo! search engine in July 2010 \cite{yahoo}. Each record in the log consists of a time stamp, user cookie id, submitted query, the top 10 results displayed to the user, and the positions of the user clicks on the returned answers. Generally speaking, typical users of Yahoo! are normal users who may not know advanced concepts, such as formal query language and schema, and use keyword queries to find their desired information. Yahoo! may generally use a combination of structured and unstructured datasets to satisfy users' intents. Nevertheless, as normal users are not aware of the existence of schema and mainly rely on the content of the returned answers to (re)formulate their queries, we expect that the users' learning mechanisms over this dataset closely resemble their learning mechanisms over structured data. We have used three different contiguous subsamples of this log whose information is shown in Table~\ref{results:table:aggstats}. The duration of each subsample is the time between the time-stamp of the first and last interaction records. Because we would like to specifically look at the users that exhibit some learning throughout their interaction, we have collected only the interactions in which a user submits at least two different queries to express the same intent. The records of the 8H-interaction sample appear at the beginning of the the 43H-interaction sample, which themselves appear at the beginning of the 101H-interaction sample.

Table~\ref{results:table:recordEx} illustrates an example of what the log record looks like. Note, for the \textit{Results} column, there are 10 results returned.  Table~\ref{results:table:scoreEx} illustrates an example of the relevance judgment scores dataset. For example, if the user entered $QueryID=00002efd$ and one of results returned was $ResultID=2722a07f$, then the relevance of that result would be three.

\begin{table}[h!]
	\centering
	\small
	\caption{Log Record Example}
	\begin{tabular}{|l|l|l|l|}
		\hline
		QueryID & CookieID & TimeStampID & Results\\
		\hline
		00002efd & 1deac14e & 1279486689 & 2722a07f ... e1468bbf\\
		12fe75b7 & 4b74d72d	& 1279546874 & dc381b38	... df9a561f\\
		\hline
	\end{tabular}
	\label{results:table:recordEx}
\end{table}

\begin{table}[h!]
	\centering
	\caption{Relevance Judgment Score Example}
	\begin{tabular}{|l|l|l|}
		\hline
		QueryID & ResultID & Score\\
		\hline
		00002efd & 2722a07f & 3\\
		12fe75b7 & dc381b38	& 4\\
		\hline
	\end{tabular}
	\label{results:table:scoreEx}
\end{table}

\subsubsection{Intent \& Reward}
Accompanying the interaction log is a set of {\it relevance judgment scores} for each query and result pair. Each relevance judgment score is a value between 0 and 4 and shows the degree of relevance of the result to the query, with 0 meaning not relevant at all and 4 meaning the most relevant result. We define the intent behind each query as the set of results with non-zero relevance scores. We use the standard ranking quality metric NDCG for the returned results of a query as the reward in each interaction as it models different levels of relevance \cite{IRBook}. The value of NDCG is between 0 and 1 and it is 1 for the most effective list.

\begin{table}[ht!]
	\centering
	\caption{Subsamples of Yahoo! interaction log}
	\begin{tabular}{|c|c|c|c|c|c|}
		\hline
		Duration & \#Interactions & \#Users & \#Queries & \#Intents \\
		\hline
		\hline
		\textasciitilde8H & 622 & 272 & 111 & 62 \\ 
		\hline
		\textasciitilde43H & 12323 & 4056 & 341 & 151\\ 
		\hline
		\textasciitilde101H & 195468 & 79516 & 13976 & 4829 \\ 
		\hline
	\end{tabular}
	\label{results:table:aggstats}
\end{table}

\subsubsection{Parameter Estimation} 
Some models, e.g., Cross's model, have some parameters that need to be trained. We have used a set of 5,000 records that appear in the interaction log immediately before the first subsample of Table~\ref{results:table:aggstats} and found the optimal values for those parameters using grid search and the sum of squared errors. 

\subsubsection{Training \& Testing} 
We train and test a single user strategy over each subsample and model, which represents the strategy of the user population in each subsample. The user strategy in each model is initialized with a uniform distribution, so that all queries are equally likely to be used for an intent. After estimating parameters, we train the user strategy using each model over 90\% of the total number of records in each selected subsample in the order by which the records appear in the interaction log. We use the value of NDCG as reward for the models that use rewards to update the user strategy after each interaction. We then test the accuracy of the prediction of using a query to express an intent for each model over the remaining 10\% of each subsample using the user strategy computed at the end of the training phase. Each intent is conveyed using only a single query in the testing portions of our subsamples. Hence, no learning is done in the testing phase and we do not update the user strategies. We report the mean squared errors over all intents in the testing phase for each subsample and model in Figure~\ref{fig:results:filtered}. A lower mean squared error implies that the model more accurately represents the users' learning method. We have excluded the Latest Reward results from the figure as they are an order of magnitude worse than the others. 

\begin{figure}[ht!]
    \centering
    \includegraphics[width = 0.9\linewidth]{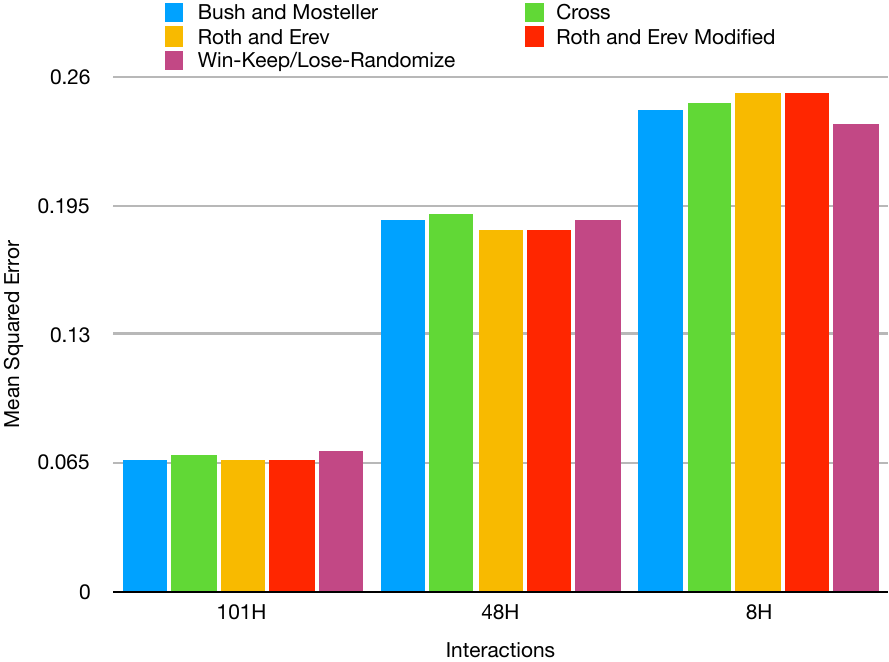}
    \caption{Accuracies of learning over the subsamples of Table~\ref{results:table:aggstats}}
    \label{fig:results:filtered}
\end{figure}

\subsubsection{Results}
Win-Keep/Lose-Randomize performs surprisingly more accurate than other methods for the 8H-interaction subsample. It indicates that in short-term and/or beginning of their interactions,  users may not have enough interactions to leverage a more complex learning scheme and use a rather simple mechanism to update their strategies. Both Roth and Erev's methods use the accumulated reward values to adjust the user strategy gradually. Hence, they cannot precisely model user learning over a rather short interaction and are less accurate than relatively more aggressive learning models such as Bush and Mosteller's and Cross's over this subsample. Both Roth and Erev's deliver the same result and outperform other methods in the 43-H and 101-H subsamples. Win-Keep/Lose-Randomize is the least accurate method over these subsamples. Since larger subsamples provide more training data, the predication accuracy of all models improves as the interaction subsamples becomes larger. The learned value for the {\it forget} parameter in the Roth and Erev's modified model is very small and close to zero in our experiments, therefore, it generally acts like the Roth and Erev's model.

Long-term communications between users and DBMS may include multiple sessions. Since Yahoo! query workload contains the time stamps and user ids of each interaction, we have been able to extract the starting and ending times of each session. Our results indicate that as long as the user and DBMS communicate over sufficiently many of interactions, e.g., about 10k for Yahoo! query workload, the users follow the Roth and Erev's model of learning. Given that the communication of the user and DBMS involve sufficiently many interactions, we have {\it not} observed any difference in the mechanism by which users learn based on the numbers of sessions in the user and DBMS communication.

\subsubsection{Conclusion} 
Our analysis indicates that users show a substantially intelligent behavior when adopting and modifying their strategies over relatively medium and long-term interactions. 
They leverage their past interactions and their outcomes, i.e., have an effective long-term memory. This behavior is most accurately modeled using Roth and Erev's model. Hence, in the rest of the paper, we set the user learning method to this model. 

%% file: 5.Theory.tex
\section{Learning Algorithm for DBMS}

\label{sec:learning}
Current systems generally assume that a user does {\it not} learn and/or modify her method of expressing intents throughout her interaction with the DBMS. However, it is known that the learning methods that are useful in static settings do not deliver desired outcomes in the dynamic ones \cite{auer2002nonstochastic}. 
Moreover, it has been shown that if the players do {\it not} use the right learning algorithms in games with identical interests, the game and its payoff may not converge to any desired states~\cite{shapley1964some}. Thus, choosing the correct learning mechanism for the DBMS is crucial to improve the payoff and converge to a desired state. The following algorithmic questions are of interest:
 \begin{enumerate}[i.]
   \item How can a DBMS learn or adapt to a user's strategy?
   \item Mathematically, is a given learning algorithm effective?
   \item What would be the asymptotic behavior of a given learning algorithm?
 \end{enumerate}
Here, we address the first and the second questions above. Dealing with the third question is far beyond the scope and space of this paper. A summary of the notations introduced in Section~\ref{sec:framework} and used in this section can be found in Table~\ref{framework:notation}.

\subsection{DBMS Reinforcement Learning}
\label{sec:arbitrary}
We adopt Roth and Erev's learning method for adaptation of the DBMS strategy, with a slight modification. The original Roth and Erev method considers only a single action space. In our work, this would translate to having only a single query. Instead we extend this such that each query has its own action space or set of possible intents. The adaptation happens over discrete time $t=0,1,2,3,\ldots$ instances where $t$ denotes the $t$th interaction of the user and the DBMS. We refer to $t$ simply as the iteration of the learning rule. 
For simplicity of notation, we refer to intent $e_i$ and result $s_\ell$ as intent $i$ and $\ell$, respectively, 
in the rest of the paper. Hence, we may rewrite the expected payoff for both user and DBMS as:
\begin{align*}
  u_r(U,D)=\sum_{i=1}^m\pi_i\sum_{j=1}^nU_{ij}\sum_{\ell=1}^{o}D_{j\ell}r_{i\ell},
\end{align*}
 where $r:[m]\times [o]\to \R^+$ is the effectiveness measure between the intent $i$ and the result, i.e., decoded intent $\ell$.
With this, the reinforcement learning mechanism for the DBMS adaptation is as follows.
\begin{enumerate}[a.]
 	\item Let $R(0)>0$ be an $n\times o$ initial reward matrix whose entries are strictly positive.
    \item Let $D(0)$ be the initial DBMS strategy with $D_{j\ell}(0)=\frac{R_{j\ell}(0)}{\sum_{\ell=1}^{o}R_{j\ell}(0)}>0$ for all $j\in [n]$ and $\ell\in [o]$.
    \item For iterations $t=1,2,\ldots$, do
    \begin{enumerate}[i.]
      \item If the user's query at time $t$ is $q(t)$, DBMS returns a result $E(t)\in E$ with probability:
      \begin{align*}
        &P(E(t)=i'\mid q(t))=D_{q(t)i'}(t).
      \end{align*}
      \item User gives a reward $r_{ii'}$ given that $i$ is the intent of the user at time $t$. Note that the reward depends both on the intent $i$ at time $t$ and the result $i'$. Then, set
      \begin{align}
        R_{j\ell}(t+1)=\left\{\begin{array}{ll}
            R_{j\ell}(t) + r_{i\ell} &\mbox{if $j=q(t)$ and $ \ell =i'$}\\
            R_{j\ell}(t)&\mbox{otherwise}
        \end{array}\right..
      \end{align}
      \item Update the DBMS strategy by
      \begin{align}\label{eqn:udaterule}
        D_{ji}(t+1)=\frac{R_{ji}(t+1)}{\sum_{\ell=1}^{o}R_{j\ell}(t+1)},
      \end{align}
      for all $j\in [n]$ and $i\in [o]$.
    \end{enumerate}
 \end{enumerate}

 In the above scheme $R(t)$ is simply the reward matrix at time $t$. 
 
 Few comments are in order regarding the above adaptation rule:
 \begin{enumerate}[-]
    \item One can use available ranking functions, e.g. \cite{Chaudhuri:2006:PIR:1166074.1166085}, for the initial reward condition $R(0)$  which possibly leads to an intuitive initial point for the learning rule. One may normalize and convert the scores returned by these functions to probability values.
    \item In step c.ii., if the DBMS has the knowledge of the user's intent after the interactions (e.g.\ through a click), the DBMS sets $R_{ji}+1$ for the known intent $i$. The mathematical analysis of both cases will be similar.
    \item In the initial step, as the DBMS uses a ranking function to compute the probabilities, it may not materialize 
    $R$ and $D$. As the game progresses, DBMS maintains the 
    strategy and reward matrices with entries for only 
    the observed queries, their underlying intents, and their returned results. 
    Hence, the DBMS does not need to materialize $R$ and $D$ for the sets of possible intents, queries, and results.
    DBMS also does not need to know the set of user's interns beforehand. Hence, the algorithm is practical for the 
    cases where the sought-for intents, 
    submitted queries, and returned results are not very large. 
    Moreover, $R$ and $D$ are generally sparse. 
    As queries and intents generally follow a power law 
    distribution \cite{IRBook}, one may use sampling techniques to 
    use this algorithm in other settings. The rigorous theoretical and empirical analysis of applying such techniques are 
    interesting subjects of future work.
\end{enumerate}
\subsection{Analysis of the Learning Rule}
We show in Section~\ref{sec:results} that users modify their strategies in data interactions. Nevertheless, ideally, one would like to use a learning mechanism for the DBMS that accurately discovers the intents behind users' queries whether or not the users modify their strategies, as it is {\it not} certain that all users will always modify their strategies. Also, in some relevant applications, the user's learning is happening in a much slower time-scale compared to the learning of the DBMS. So, one can assume that the user's strategy is fixed compared to the time-scale of the DBMS adaptation. Therefore, first, we consider the case that the user is {\it not} adapting her strategy, i.e., she has a fixed strategy during the interaction. Then, we consider the case that the user's strategy is adapting to the DBMS's strategy but perhaps on a slower time-scale in Section~\ref{sec:userlearning}.

In this section, we provide an analysis of the reinforcement mechanism provided above and will show that, statistically speaking, the adaptation rule leads to improvement of the efficiency of the interaction. Note that since the user gives feedback only on one tuple in the result, one can without the loss of generality assume that the cardinality of the list $k$ is $1$.

For the analysis of the reinforcement learning mechanism in Section~\ref{sec:learning} and for simplification, denote
\begin{align}\label{eqn:ut}
  u(t):=u_r(U,D(t))=u_r(U,D(t)),
\end{align}
for an effectiveness measure $r$ as $u_r$
is defined in \eqref{eqn:payoff}.

We recall that a random process $\{X(t)\}$ is a submartingale \cite{durrett2010} if it is absolutely integrable (i.e.\ $E(|X(t)|)<\infty$ for all $t$) and
\begin{align*}
  E(X(t+1)\mid \F_t)\geq X(t),
\end{align*}
where $\F_t$ is the history or $\sigma$-algebra generated by $X_1,\ldots,X_t$. In other words, a process $\{X(t)\}$ is a sub-martingale if the expected value of $X(t+1)$ given $X(t),X(t-1),\ldots, X(0)$, is not strictly less than the value of $X_t$. Note that submartingales are nothing but the stochastic counterparts of monotonically increasing sequences. As in the case of bounded (from above) monotonically increasing sequences, submartingales pose the same property, i.e.\ any submartingale $\{X(t)\}$ with $E(|X(t)|)<B$ for some $B\in \R^+$ and all $t\geq 0$ is convergent almost surely. We refer the interested readers to \cite{durrett2010} for further information on this result (martingale convergence theorem).

The main result in this section is that the sequence of the utilities $\{u(t)\}$ (which is indeed a stochastic process as $\{D(t)\}$ is a stochastic process) defined by \eqref{eqn:ut} is a submartignale when the reinforcement learning rule in Section~\ref{sec:learning} is utilized. As a result the proposed reinforcement learning rule \textit{stochastically} improves the efficiency of communication between the DBMS and the user. More importantly, this holds \textit{for an arbitrary reward/effectiveness measure $r$}. This is rather a very strong result as the algorithm is  robust to the choice of the reward mechanism.

 To show this, we discuss an intermediate result. For simplicity of notation, we fix the time $t$ and we use superscript $+$ to denote variables at time $(t+1)$ and drop the dependencies at time $t$ for variables depending on time $t$. Throughout the rest of our discussions, we let $\{\F_t\}$ be the natural filtration for the process $\{
D(t)\}$, i.e.\ $\F$ is the $\sigma$-algebra generated by $D(0),\ldots,D(t)$.

\begin{lemma}\label{lemma:expected}
  For any $\ell \in [m]$ and $j\in [n]$, we have
  \begin{align*}
    &E(D^+_{j\ell}\mid \F_t)-D_{j\ell}\\
    &\quad=D_{j\ell} \cdot \sum_{i=1}^m \pi_i U_{ij}\left( \frac{r_{i\ell}}{ \bar{R}_j+r_{il}} - \sum_{\ell' =1}^{o}  D_{j\ell'}\frac{r_{i\ell'}}{\bar{R}_j +r_{i\ell'}}\right),
  \end{align*}
  where
  $\bar{R}_j = \sum_{\ell' =1}^{o} R_{j\ell'}.$
\end{lemma}
\begin{proof}
  Fix $\ell\in [m]$ and $j\in [n]$. Let $A$ be the event that at the $t$'th iteration, we reinforce a pair $(j,\ell')$ for some $\ell' \in [m]$. Then on the complement $A^c$ of $A$, $D^+_{j\ell}(\omega)=D_{j\ell}(\omega)$. Let $A_{i,\ell'} \subseteq A$ be the subset of $A$ such that the intent of the user is $i$ and the pair $(j,\ell')$ is reinforced. Note that the collection of sets $\{A_{i,\ell'}\}$ for  $i, \ell' \in [m],$ are pairwise mutually exclusive and their union constitute the set $A.$

  We note that
  \begin{align*}
    D^+_{j\ell}&=\sum_{i=1}^m \left(\frac{R_{j\ell}+r_{il}}{\bar{R}_{j}+r_{i\ell}}1_{A_{i,\ell}}+
    \sum_{\substack{\ell' =1 \\ \ell'\neq \ell}}^{{o}} \frac{R_{j\ell}}{\bar{R}_{j}+ r_{i \ell'}}1_{A_{i,\ell'}}\right)\cr
    &\qquad +
    D_{j\ell}1_{A^c}.
  \end{align*}
  Therefore, we have
  \begin{align*}
    &E(D^+_{j\ell}\mid \F_t)= \sum_{i=1}^{m} \pi_iU_{ij}D_{j\ell}\frac{R_{j\ell}+r_{i\ell}}{\bar{R}_{j}+r_{i\ell}}\cr
    &+\sum_{i=1}^{m} \pi_iU_{ij}\sum_{\ell \not=\ell'}D_{j\ell' }\frac{R_{j\ell}}{\bar{R}_{j}+r_{i\ell'}}+
    (1-p)D_{j\ell},
  \end{align*}
  where $p=\mathbb{P}(A\mid \F)$. Note that $D_{j\ell}=\frac{R_{ji}}{\bar{R}_j}$ and hence,
  \begin{align*}
    E(D^+_{j\ell}\mid \F_t)-D_{j\ell}&=\sum_{i=1}^m \pi_iU_{ij}D_{j\ell}\frac{r_{i\ell}\bar{R}_j - R_{j\ell}}{\bar{R}_j(\bar{R}_j +r_{i\ell})} \cr
    &\qquad-\sum_{i=1}^m \pi_iU_{ij}\sum_{\ell\neq \ell'} D_{j\ell'} \frac{R_{j\ell} r_{i\ell'}}{\bar{R}_j(\bar{R}_j +r_{i\ell'})}.
  \end{align*}
  Replacing $\frac{R_{jl}}{\bar{R}_j}$ with $D_{j\ell}$ and rearranging the terms in the above expression, we get the result.
\end{proof}
To show the main result, we will utilize the following result in martingale theory.
\begin{theorem}\label{thrm:convergence}\cite{robbins1985convergence}
A random process $\{X_t\}$ converges almost surely if $X_t$ is bounded, i.e., $E(|X_t|) <B$ for some $B\in \R^+$ and all $t\geq 0$ and
\begin{align}
E(X_{t+1}| \F_t)  \geq X_t - \beta_t
\end{align}
where $\beta_t\geq 0$ is a summable sequence almost surely, i.e., $\sum_t \beta_t < \infty$ with probability~$1.$
\end{theorem}
Note that this result is a weaker form of the Robins-Siegmund martingale convergence theorem in {\cite{robbins1985convergence}} but it will serve for the purpose of our discussion.

Using Lemma~\ref{lemma:expected} and the above result, we show that up to a summable disturbance, the proposed learning mechanism is stochastically improving.
\begin{theorem}\label{thrm:submartingale}
  Let $\{u(t)\}$ be the sequence given by \eqref{eqn:ut}. Then,
  \[E(u(t+1\mid \F_t)\geq E(u(t)\mid \F_t)-\beta_t,\]
  for some non-negative random process $\{\beta_t\}$ that is summable (i.e.\ $\sum_{t=0}^\infty\beta<\infty$ almost surely). As a result $\{u(t)\}$ converges almost surely.
\end{theorem}
\begin{proof}
Let $u^+:=u(t+1)$, $u:=u(t)$,
\[u^j:=u^j(U(t),D(t))=\sum_{i=1}^m\sum_{\ell=1}^{o}\pi_iU_{ij}D_{j\ell}r_{i\ell(t)},\]
and also define $\bar{R}_j:=\sum_{\ell'=1}^mR_{j\ell'}$. Note that $u^j$ is the efficiency of the $j$th signal/query.

Using the linearity of conditional expectation and Lemma~\ref{lemma:expected}, we have:
\begin{align}\label{eqn:submartingale}
    &E(u^+\mid \F_t)-u=\sum_{i=1}^m\sum_{j=1}^n\pi_iU_{ij} \sum_{\ell=1}^{o} r_{i\ell'} \left(E(D^+_{j\ell}\mid \F_t)-D_{j\ell}\right)\cr
    &=\sum_{i=1}^m\sum_{j=1}^n\sum_{\ell=1}^{o} \pi_iU_{ij} D_{j\ell}r_{i\ell}\left( \sum_{i'=1}^m \pi_i' U_{i'j}\left( \frac{r_{i'\ell}}{ \bar{R}_j+r_{i'\ell}} \right.\right. \cr
    & \left.\left. \qquad\qquad\qquad  -\sum_{\ell' =1}^{o}  D_{j\ell'}\frac{r_{i'\ell'}}{\bar{R}_j +r_{i'\ell'}}\right)\right).
 \end{align}
 Now, let $y_{j\ell} = \sum_{i=1}^m \pi_i U_{ij}r_{i\ell}$ and $z_{j\ell} = \sum_{i=1}^m \pi_i U_{ij} \frac{r_{i\ell}}{ \bar{R}_j+r_{i\ell}}$. Then, we get from the above expression that
 \begin{align}\label{eqn:submartingaleshort}
    &E(u^+\mid \F_t)-u =\cr
    &\qquad \sum_{j=1}^n\left(\sum_{\ell=1}^{o} D_{j\ell} y_{i\ell}z_{j\ell}-
     \sum_{\ell=1}^{o} D_{j\ell}y_{j\ell}\sum_{\ell' =1}^{o}  D_{j\ell'}z_{j\ell'}\right).
\end{align}

Now, we express the above expression as
\begin{equation}\label{eqn:submalast}
E(u^+\mid \F_t)-u = V_t + \tilde{V}_t
\end{equation}
where
\[V_t = \sum_{j=1}^n \frac{1}{\bar{R}_j} \left(\sum_{\ell=1}^{o} D_{j\ell} y_{j\ell}^2 - \left(\sum_{l=1}^{o} D_{j\ell}y_{j\ell}\right)^2\right),\]
and
\begin{align}\label{eqn:tildeV}
\tilde{V}_t = \sum_{j=1}^n\left( \sum_{\ell=1}^{o} D_{j\ell} y_{j\ell} \sum_{\ell' =1}^{o} D_{j\ell'}\tilde{z}_{j\ell'} - \sum_{\ell=1}^m D_{j\ell} y_{j\ell} \tilde{z}_{j\ell} \right).
\end{align}
Further,  $\tilde{z}_{j\ell} = \sum_{i=1} \pi_i U_{ij} \frac{r_{i\ell}^2}{\bar{R}_j (\bar{R}_j+r_{i\ell})}.$

We claim that $V_t \geq 0$ for each $t$ and $\{\tilde{V}_t\}$ is a summable sequence almost surely. Then, from \eqref{eqn:submalast} and Theorem~\ref{thrm:convergence}, we get that $\{u_t\}$ converges almost surely and it completes the proof. Next, we validate our claims.

We first show that $V_t \geq 0, \forall t.$ Note that $D$ is a row-stochastic matrix and hence, $\sum_{\ell=1}^{o}D_{j\ell}=1$. Therefore, by the Jensen's inequality \cite{durrett2010}, we have:
\begin{align*}
\sum_{\ell=1}^{o}D_{j\ell}(y_{j\ell})^2\geq \sum_{\ell=1}^{o}(D_{j\ell}y_{j\ell})^2.
\end{align*}
Hence, $V \geq 0.$

We next claim that $\{\tilde{V}_t\}$ is a summable sequence with probability one. It can be observed from \eqref{eqn:tildeV} that
\begin{equation}
V_t \leq \sum_{j=1}^{o} \frac{o^2n}{\bar{R}_j^2}.
\end{equation}
since $y_{j\ell} \leq 1, \tilde{z}_{j\ell} \leq \bar{R}_j^{-2}$ for each $j \in [n], \ell \in [m]$ and $D$ is a row-stochastic matrix. To prove the claim, it suffices to show that for each $j \in [m]$, the sequence $\{\frac{1}{R^2_{j}(t)}\}$ is summable. Note that for each $j \in [m]$ and for each $t$, we have $\bar{R}_j(t+1) = \bar{R}_j(t) + \epsilon_t$ where $\epsilon_t \geq \epsilon >0$ with probability $p_t \geq p >0$. Therefore, using the Borel-Cantelli Lemma for adapted processes \cite{durrett2010} we have $\{\frac{1}{R^2_{j}(t)}\}$ is summable which concludes the proof.
\end{proof}

The above result implies that the effectiveness of the DBMS, stochastically speaking, increases as time progresses when the learning rule in Section~\ref{sec:learning} is utilized. Not only that, but this property does not depend on the choice of the effectiveness function (i.e.\ $r_{i\ell}$ in this case). This is indeed a desirable property for any adapting scheme for DBMS adaptation.
\subsection{User Adaptation}
 \label{sec:userlearning}
Here, we consider the case that the user also adapts to the DBMS's strategy. 
When the user submits a query $q$ and the DBMS returns a result 
that fully satisfy the intent behind the query $e$, 
it is relatively more likely that the user will use the query $q$ 
to express $e$ again during her interaction with the DBMS. 
On the other hand, if the DBMS returns a result that does not
contain any tuple that is relevant to the $e$, 
it less likely that the user expresses $e$ using $q$ in future.
In fact, researchers have observed that users show reinforcement 
learning behavior when interacting 
with a DBMS over a period of time~\cite{cen2013reinforcement}. 
In particular, the authors in \cite{cen2013reinforcement} have shown
that some groups of users learned to formulate queries with a model 
similar to Roth-Erev reinforcement learning.
We define the similarity measure as follows. For simplicity we assume that $m=o$ and use the following similarity measure:

 
 \begin{align*}
   r_{i\ell}=\left\{\begin{array}{ll}
                      1 & \mbox{if $i=\ell$},  \\
                      0 & \mbox{otherwise}
                    \end{array}\right..
 \end{align*}

  In this case, we assume that the user adapts to the DBMS strategy at time steps $0<t_1<\cdots<t_k<\cdots$ and in those time-steps the DBMS is not adapting as there is no reason to assume the synchronicity between the user and the DBMS. The reinforcement learning mechanism for the user is as follows:
 \begin{enumerate}[a.]
  \item Let $S(0)>0$ be an $m\times n$ reward matrix whose entries are strictly positive.
    \item Let $U(0)$ be the initial user's strategy with $$U_{ij}(0)=\frac{S_{ij}(0)}{\sum_{j'=1}^{n}S_{ij'}(0)}$$ for all $i\in [m]$ and $j\in [n].$ and let $U(t_k)=U(t_k-1)=\cdots=U(t_{k-1}+1)$ for all $k$.
     \item For all $k \geq 1,$ do the following:
    \begin{enumerate}[i.]
      \item The nature picks a random intent $t \in [m]$ with probability $\pi_i$ (independent of the earlier choices of the nature) and the user picks a query $j\in [n]$ with probability
      \begin{align*}
        &P(q(t_k)=j\mid i(t_k)=i)=U_{ij}(t_k).
      \end{align*}
      \item The DBMS uses the current strategy $D(t_k)$ and interpret the query by the intent $i'(t)=i'$ with probability
      \begin{align*}
      &U(i'(t_k)=i'\mid q(t_k)=j)=D_{ji'}(t_k).
      \end{align*}
      
      \item User gives a reward $1$ if $i=i'$ and otherwise, gives no rewards, i.e. 
      \begin{align*}
        S_{ij}^+=\left\{\begin{array}{ll}
            S_{ij}(t_k) + 1 &\mbox{if $j=q(t_k)$ and $i(t_k)=i'(t_k)$}\\
            S_{ij}(t_k)&\mbox{otherwise}
        \end{array}\right.
      \end{align*}
      where $S_{ij}^+= S_{ij}(t_k+1)$. 
      \item Update the user's strategy by
      \begin{align}\label{eqn:udateruleuser}
        U_{ij}(t_k+1)=\frac{U_{ji}(t_k+1)}{\sum_{j'=1}^{n}S_{ij'}(t_k+1)},
      \end{align}
      for all $i\in [m]$ and $j\in [n]$.
    \end{enumerate}
 \end{enumerate}

 In the above scheme $S(t)$ is the reward matrix at time $t$ for the user. 

\subsection{Analysis of User and DBMS Adaptation}
In this section, we provide an analysis of the reinforcement mechanism provided above and will show that, statistically speaking, 
our proposed adaptation rule for DBMS, even when the user adapts, leads to improvement of the efficiency of the interaction. 
With a slight abuse of notation, let 
\begin{align}\label{eqn:utnew}
  u(t):=u_r(U,D(t))=u_r(U(t),D(t)),
\end{align}
for an effectiveness measure $r$ as $u_r$
is defined in \eqref{eqn:payoff}.

\begin{lemma} 
Let $t = t_k$ for some~$k \in \mathbb{N}.$ Then, for any $i \in [m]$ and $j\in [n],$ we have 
\begin{align}\label{eq:lem}
E(U_{ij}^+ \mid \F_{t})  - U_{ij} 
= \frac{\pi_i U_{ij}}{\sum_{\ell=1}^n S_{i\ell}+1} (D_{ji} - u^i(t))
\end{align}
where 
\begin{equation*}
u^i(t) = \sum_{j=1}^n U_{ij}(t) D_{ji}(t).
\end{equation*}
\end{lemma}
\begin{proof}
Fix $i\in [m], j\in [n]$ and $k \in \mathbb{N}$. Let $B$ be the event that at the $t_k$'th iteration, user reinforces a pair $(i,\ell)$ for some $\ell \in [n]$. Then, on the complement $B^c$ of $B$, $P^+_{ij}(\omega)=P_{ij}(\omega)$. Let $B_1\subseteq B$ be the subset of $B$ such that the pair $(i,j)$ is reinforced and $B_2=B \setminus B_1$ be the event that some other pair $(i, \ell)$ is reinforced for $\ell\not=i$.

  We note that
  \begin{align*}
    U^+_{ij}=\frac{S_{ij}+1}{\sum_{\ell=1}^n S_{i\ell}+1}1_{B_1}+
    \frac{S_{ij}}{\sum_{\ell=1}^n S_{i\ell}+1}1_{B_2}+
    U_{ij}1_{B^c}.
  \end{align*}
  Therefore, we have
  \begin{align*}
    &E(U^+_{ij}\mid \F_{k_t})=\pi_iU_{ij}D_{ji}\frac{S_{ij}+1}{\sum_{\ell=1}^nS_{i\ell}+1}\cr
    &+\sum_{\ell\not=j}\pi_{i}U_{i\ell}D_{\ell i}\frac{S_{ij}}{\sum_{\ell'=1}^nS_{i\ell'}+1}+
    (1-p)U_{ij},
  \end{align*}
  where $p=U(B \mid \F_{k_t}) = \sum_{\ell} \pi_i U_{ij} D_{ji}$. Note that $U_{ij}=\frac{S_{ij}}{\sum_{\ell=1}^n S_{i\ell}}$ and hence, 
  \begin{align*}
    E(U^+_{ij}\mid \F_t) & -U_{ij}= \cr
    &\frac{1}{\sum_{\ell'=1}^nS_{i\ell'}+1}\left(\pi_iU_{ij}D_{ji} - \pi_i U_{ij} \sum_{\ell} U_{i\ell}D_{\ell i}  \right).
  \end{align*}
  which can be rewritten as in \eqref{eq:lem}.
\end{proof}

Using Lemma~\ref{lemma:expected}, we show that the process $\{u(t)\}$ is a sub-martingale.
\begin{theorem}\label{thrm:submartingalenew}
Let $t = t_k$ for some $k \in \mathbb{N}.$ Then, we have
 \begin{align}
 E(u(t+1) \mid \F_t) - u(t) \geq 0
 \end{align} 
where $u(t)$ is given by \eqref{eqn:utnew}. 
\end{theorem}
\begin{proof}
Fix $t = t_k$ for some $ k\in \mathbb{N}.$ Let $u^+:=u(t+1)$, $u:=u(t)$, $u^i:=u^i(U(t),D(t))$ and also define $\tilde{S}^i:={\sum_{\ell'=1}^mS_{i\ell'}+1}$. Then, using the linearity of conditional expectation and Lemma~\ref{lemma:expected}, we have:
\begin{align}\label{eqn:submartingale:2}
    &E(u^+\mid \F_t)-u=\sum_{i=1}^m\sum_{j=1}^n\pi_iD_{ji}\left(E(U^+_{ij}\mid \F_t)-U_{ij}\right)\cr
    &=\sum_{i=1}^m\sum_{j=1}^n\pi_iD_{ji} \frac{\pi_iU_{ij}}{\sum_{\ell'=1}^mS_{j\ell'}+1}\left(D_{ji}-u^i\right)\cr
    &=\sum_{i=1}^m\frac{\pi_i^2}{\tilde{S}^i}\left(\sum_{j=1}^n U_{ij}(D_{ji})^2-(u^i)^2\right).
\end{align}
Note that $U$ is a row-stochastic matrix and hence, $\sum_{i=1}^mU_{ij}=1$. Therefore, by the Jensen's inequality \cite{durrett2010}, we have:
\begin{align*}
\sum_{j=1}^n U_{ij}(D_{ji})^2\geq \left(\sum_{j=1}^nD_{ji}U_{ij}\right)^2=(u^i)^2.
\end{align*}
Replacing this in the right-hand-side of \eqref{eqn:submartingale:2}, we conclude that $E(u^+\mid \F_t)-u\geq 0$ and hence, the sequence $\{u(t)\}$ is a submartingale.
\end{proof}

\begin{corollary}
  The sequence $\{u(t)\}$ given by \eqref{eqn:ut} converges almost surely.
\end{corollary}
\begin{proof}
  Note from Theorem~\ref{thrm:submartingale} and \ref{thrm:submartingalenew} that the sequence $\{u(t)\}$ satisfies all the conditions of Theorem~\ref{thrm:convergence}.  Hence, proven.
\end{proof}

%% file: 6.Relational.tex
\section{Efficient Query Answering over Relational Databases}
\label{sec:relational}
An efficient implementation of the algorithm proposed in Section~\ref{sec:learning} over large relational databases poses two challenges. First, since the set of possible interpretations and their results for a given query is enormous, one has to find efficient ways of maintaining users' reinforcements and updating DBMS strategy. Second, keyword and other usable query interfaces over databases normally return the top-$k$ tuples according to some scoring functions~\cite{IRStyle,Chen:2009:KSS}. Due to a series of seminal works by database researchers~\cite{Fagin:2001:OAA:375551.375567}, there are efficient algorithms to find such a list of answers. Nevertheless, our reinforcement learning algorithm uses a randomized semantic for answering algorithms in which candidate tuples are associated a probability for each query that reflects the likelihood by which it satisfies the intent behind the query. The tuples must be returned randomly according to their associated probabilities. Using (weighted) sampling to answer SQL queries with aggregation functions approximately and efficiently is an active research area \cite{DBLP:conf/sigmod/ChaudhuriDK17,Idreos:SIGMOD:2015}. However, there has not been any attempt on using a randomized strategy to answer so-called point queries over relational data and achieve a balanced exploitation-exploration trade-off efficiently.

\subsection{Maintaining DBMS Strategy}
\label{sec:system}
\subsubsection{Keyword Query Interface}
\label{sec:keywordQueryInterface}
We use the current architecture of keyword query interfaces over relational databases that directly use schema information to interpret the input keyword query \cite{Chen:2009:KSS}. A notable example of such systems is IR-Style \cite{IRStyle}. As it is mentioned in Section~\ref{sec:DBStrategy}, given a keyword query, these systems translate the input query to a Select-Project-Join query whose {\it where} clause contains function $match$. The results of these interpretations are computed, scored according to some ranking function, and are returned to the user. We provide an overview of the basic concepts of such a system. We refer the reader to \cite{IRStyle,Chen:2009:KSS} for more explanation.

\subsubsection{Tuple-set:} Given keyword query $q$, a {\it tuple-set} is a set of tuples in a base relation that contain some terms in $q$. After receiving $q$, the query interface uses an inverted index to compute a set of tuple-sets. For instance, consider a database of products with relations {\it Product(pid, name)}, {\it Customer(cid, name)}, and {\it ProductCustomer(pid, cid)} where {\it pid} and {\it cid} are numeric strings. Given query {\it iMac John}, the query interface returns a tuple-set from {\it Product} and a tuple-set from {\it Customer} that match at least one term in the query. The query interface may also use a scoring function, e.g., traditional TF-IDF text matching score, to measure how exactly each tuple in a tuple-set matches some terms in $q$.

\subsubsection{Candidate Network:} A {\it candidate network} is a join expression that connects the tuple-sets via primary key-foreign key relationships. A candidate network joins the tuples in different tuple-sets and produces joint tuples that contain the terms in the input keyword query. One may consider the candidate network as a join tree expression whose leafs are tuple-sets. For instance, one candidate network for the aforementioned database of products is {\it Product} $\bowtie$ {\it ProductCustomer} $\bowtie$ {\it Customer}. To connect tuple-sets via primary key-foreign key links, a candidate network may include base relations whose tuples may not contain any term in the query, e.g., {\it ProductCustomer} in the preceding example. Given a set of tuple-sets, the query interface uses the schema of the database and progressively generates candidate networks that can join the tuple-sets. For efficiency considerations, keyword query interfaces limit the number of relations in a candidate network to be lower than a given threshold. For each candidate network, the query interface runs a SQL query and return its results to the users.There are algorithms to reduce the running time of this stage, e.g., run only the SQL queries guaranteed to produce top-$k$ tuples \cite{IRStyle}. Keyword query interfaces normally compute the score of joint tuples by summing up the scores of their constructing tuples multiplied by the inverse of the number of relations in the candidate network to penalize long joins. We use the same scoring scheme. We also consider each (joint) tuple to be candidate answer to the query if it contains at least one term in the query.

\subsubsection{Managing Reinforcements}
The aforementioned keyword query interface implements a basic DBMS strategy of mapping queries to results but it does not leverage users' feedback and adopts a deterministic strategy without any exploration. A naive way to record users' reinforcement is to maintain a mapping from queries to tuples and directly record the reinforcements applied to each pair of query and tuple. In this method, the DBMS has to maintain the list of all submitted queries and returned tuples. Because many returned tuples are the joint tuples produced by candidate networks, it will take an enormous amount of space and is inefficient to update. Hence, instead of recording reinforcements directly for each pair of query and tuple, we construct some features for queries and tuples and maintain the reinforcement in the constructed feature space. More precisely, we construct and maintain a set of {\it n-gram} features for each attribute value in the base relations and each input query. N-grams are contiguous sequences of terms in a text and are widely used in text analytics and retrieval \cite{IRBook}. In our implementation, we use up to 3-gram features to model the challenges in managing a large set of features. Each feature in every attribute value in the database has its associated attribute and relation names to reflect the structure of the data. We maintain a reinforcement mapping from query features to tuple features. After a tuple gets reinforced by the user for an input query, our system increases the reinforcement value for the Cartesian product of the features in the query and the ones in the reinforced tuple. According to our experiments in Section~\ref{sec:comparison}, this reinforcement mapping can be efficiently maintained in the main memory by only a modest space overhead. 

Given an input query $q$, our system computes the score of each tuple $t$ in every tuple-set using the reinforcement mapping: it finds the n-gram features in $t$ and $q$ and sums up their reinforcement values recorded in the reinforcement mapping. Our system may use a weighted combination of this reinforcement score and traditional text matching score, e.g., TF-IDF score, to compute the final score. One may also weight each tuple feature proportional to its inverse frequency in the database similar to some traditional relevance feedback models \cite{IRBook}. Due to the space limit, we mainly focus on developing an efficient implementation of query answering based on reinforcement learning over relational databases and leave using more advanced scoring methods for future work. The scores of joint tuples are computed as it is explained in Section~\ref{sec:keywordQueryInterface}. We will explain in Section~\ref{sec:sampling}, how we convert these scores to probabilities and return tuples. Using features to compute and record user feedback has also the advantage of using the reinforcement of a pair of query and tuple to compute the relevance score of other tuples for other queries that share some features. Hence, reinforcement for one query can be used to return more relevant answers to other queries.

\subsection{Efficient Exploitation \& Exploration}
\label{sec:sampling}
We propose the following two algorithms to generate a weighted random sample of size $k$ over all candidate tuples for a query.

\subsubsection{Reservoir}
\label{app:res} 
To provide a random sample, one may calculate the total scores of all candidate answers to compute their sampling probabilities. Because this value is not known beforehand, one may use weighted reservoir sampling \cite{Chaudhuri:1999:RSO:304182.304206} to deliver a random sample without knowing the total score of candidate answers in a single scan of the data as follows. 

\begin{algorithm}
\caption{Reservoir}\label{resevoir}
  \begin{algorithmic}[]
  \State $W \gets 0$
  \State $\text{Initialize reservoir array }A[k] \text{to }k \text{dummy tuples.}$
  \ForAll{$\text{candidate network }CN$}
    \ForAll{$t \in CN$}
      \If{$ A \text{ has dummy values}$}
        \State $\text{insert }k \text{ copies of }t\text{ into }A$
      \Else
        \State $W \leftarrow$ $W + Sc(t)$
        \ForAll{$i=1 \in k$}
          \State $\text{insert }t \text{ into }A[i] \text{ with probability }\frac{Sc(t)}{W}$
        \EndFor
      \EndIf
    \EndFor
  \EndFor
  \end{algorithmic}
\end{algorithm}

Reservoir generates the list of answers only after computing the results of all candidate networks, therefore, users have to wait for a long time to see any result. It also computes the results of all candidate networks by performing their joins fully, which may be inefficient. We propose the following optimizations to improve its efficiency and reduce the users' waiting time. 

\subsubsection{Poisson-Olken} 
{\it Poisson-Olken} algorithm uses Poisson sampling to output progressively the selected tuples as it processes each candidate network. It selects the tuple $t$ with probability $\frac{Sc(t)}{M}$, where $M$ is an upper bound to the total scores of all candidate answers. To compute $M$, we use the following heuristic. Given candidate network $CN$, we get the upper bound for the total score of all tuples generated from $CN$: $M_{CN}=$ $\frac{1}{n}( \sum_{TS \in CN} Sc_{max}(TS))$ $\frac{1}{2}\Pi_{TS \in CN} \card{TS}$ in which $Sc_{max}(TS)$ is the maximum query score of tuples in the tuple-set $TS$ and $\card{TS}$ is the size of each tuple-set. The term $\frac{1}{n} ( \sum_{TS \in CN} Sc_{max}(TS))$ is an upper bound to the scores of tuples generated by $CN$. Since each tuple generated by $CN$ must contain one tuple from each tuple-set in $CN$, the maximum number of tuples in $CN$ is $\Pi_{TS \in CN} \card{TS}$. It is very unlikely that all tuples of every tuple-set join with all tuples in every other tuple-set in a candidate network. Hence, we divide this value by $2$ to get a more realistic estimation. We do not consider candidate networks with cyclic joins, thus, each tuple-set appears at most once in a candidate network. The value of $M$ is the sum of the aforementioned values for all candidate networks with size greater than one and the total scores of tuples in each tuple-set. Since the scores of tuples in each tuple-set is kept in the main memory, the maximum and total scores and the size of each tuple-set is computed efficiently before computing the results of any candidate network.   

Both {\it Reservoir} and the aforementioned Poisson sampling compute the full joins of each candidate network and then sample the output. This may take a long time particularly for candidate networks with some base relations. There are several join sampling methods that compute a sample of a join by joining only samples the input tables  and avoid computing the full join \cite{Olken93randomsampling,Chaudhuri:1999:RSO:304182.304206,DBLP:conf/sigmod/KandulaSVOGCD16}. To sample the results of join $R_1 \bowtie R_2$, most of these methods must know some statistics, such as the number of tuples in $R_2$ that join with each tuple in $R_1$, before performing the join. They precompute these statistics in a preprocessing step for each base relation. But, since $R_1$ and/or $R_2$ in our candidate networks may be tuples sets, one cannot know the aforementioned statistics unless one performs the full join. 

However, the join sampling algorithm proposed by Olken \cite{Olken93randomsampling} finds a random sample of the join without the need to precompute these statistics. Given join $R_1 \bowtie R_2$, let $t \rtimes R_2$ denote the set of tuples in $R_2$ that join with $t \in R_1$, i.e., the right semi-join of $t$ and $R_2$. Also, let $\card{t\rtimes R_2}^{t \in R_1}_{max}$ be the maximum number of tuples in $R_2$ that join with a single tuple $t \in R_1$. The Olken algorithm first randomly picks a tuple $t_1$ from $R_1$. It then randomly selects the tuple $t_2$ from $t_1 \rtimes R_2$. It accepts the joint tuple $t_1 \bowtie t_2$ with probability $\frac{\card{t_1 \rtimes R_2}}{\card{t\rtimes R_2}^{t \in R_1}_{max}}$ and rejects it with the remaining probability. To avoid scanning $R_2$ multiple times, Olken algorithm needs an index over $R_2$. Since the joins in our candidate networks are over only primary and foreign keys, we do not need too many indexes to implement this approach.

We extend the Olken algorithm to sample the results of a candidate network without doing its joins fully as follows. Given candidate network $R_1 \bowtie R_2$, our algorithm randomly samples tuple $t_1 \in R_1$ with probability $\frac{Sc(t_1)}{\sum_{t \in R_1}{(Sc(t))}}$, where $Sc(t)$ is the score of tuple $t$, if $R_1$ is a tuple-set. Otherwise, if $R_1$ is a base relation, it picks the tuple with probability $\frac{1}{\card{R_1}}$. The value of $\sum_{t \in R}{(Sc(t))}$ for each tuple set $R$ is computed at the beginning of the query processing and the value of ${\card{R}}$ for each base relation is calculated in a preprocessing step. The algorithm then samples tuple $t_2$ from $t_1 \rtimes R_2$ with probability $\frac{Sc(t_2)}{\sum_{t \in t_1 \rtimes R_2}{(Sc(t))}}$ if $R_2$ is a tuple-set and $\frac{1}{\card{t_1 \rtimes R_2}}$ if $R_2$ is a base relation. It accepts the joint tuple with probability $\frac{\sum_{t \in t_1 \rtimes R_2}Sc(t)}{\max{(\sum_{t \in s \rtimes R_2, s \in R_1}Sc(t))}}$ and rejects it with the remaining probability.

To compute the exact value of $\max{(\sum_{t \in s \rtimes R_2, s \in R_1}Sc(t))}$, one has to perform the full join of $R_1$ and $R_2$. Hence, we use an upper bound on $\max{(\sum_{t \in s \rtimes R_2, s \in R_1}Sc(t))}$ in Olken algorithm. Using an upper bound for this value, Olken algorithm produces a correct random sample but it may reject a larger number of tuples and generate a smaller number of samples. To compute an upper bound on the value of $\max{(\sum_{t \in s \rtimes R_2, s \in R_1}Sc(t))}$, we precompute the value of $\card{t\rtimes B_i}^{t \in B_j}_{max}$ before the query time for all base relations $B_i$ and $B_j$ with primary and foreign keys of the same domain of values. Assume that $B_1$ and $B_2$ are the base relations of tuple-sets $R_1$ and $R_2$, respectively. We have $\card{t\rtimes R_2}^{t \in R_1}_{max}$ $ \leq \card{t\rtimes B_2}^{t \in B_1}_{max}$. Because $\max{(\sum_{t \in s \rtimes R_2, s \in R_1}Sc(t))}$ $ \leq \max_{t \in R_2}{(Sc(t))} \card{t\rtimes R_2}^{t \in R_1}_{max}$, we have $\max{(\sum_{t \in s \rtimes R_2, s \in R_1}Sc(t))}$ $\leq \max_{t \in R_2}{(Sc(t))} \card{t\rtimes B_2}^{t \in B_1}_{max}$. Hence, we use $\frac{\sum_{t \in t_1 \rtimes R_2}Sc(t)}{\max_{t \in R_2}{(Sc(t))} \card{t\rtimes B_2}^{t \in B_1}_{max}}$ for the probability of acceptance. We iteratively apply the aforementioned algorithm to candidate networks with multiple joins by treating the join of each two relations as the first relation for the subsequent join in the network. 

The following algorithm adopts a Poisson sampling method to return a sample of size $k$ over all candidate networks using the aforementioned join sampling algorithm. We show binomial distribution with parameters $n$ and $p$ as $B(n,p)$.
We denote the aforementioned join algorithm as {\it Extended-Olken}. Also, $ApproxTotalScore$ denotes the approximated value of total score computed as explained at the beginning of this section.

\begin{algorithm}[ht!]
\caption{Poisson-Olken}\label{poissonolken}
  \begin{algorithmic}[]
  \State $x \gets k$
  \State $W \gets \frac{ApproxTotalScore}{k}$
  \While {$x > 0$}
    \ForAll{$\text{ candidate network } CN$}
      \If{$CN \text{ is a single tuple-set }$}
        \ForAll{$t \in CN$}
          \State $\text{output } t \text{ with probability } \frac{Sc(t)}{W}$
          \If{$ \text{a tuple } t \text{ is picked}$}
            \State $ x \gets x-1$
          \EndIf
        \EndFor
      \Else
        \State $\text{let } CN = R_1\bowtie \ldots \bowtie R_n$
        \ForAll{$t \in R_1$}
          \State $\text{Pick value } X \text{ from distribution } B (k, \frac{Sc(t)}{W})$
          \State $\text{Pipeline } X \text{ copies of }t \text{ to the Olken algorithm }$
          \If{$\text{Olken accepts } m \text{ tuples }$}
            \State $x \leftarrow x-m$
          \EndIf
       \EndFor
      \EndIf
    \EndFor
  \EndWhile
  \end{algorithmic}
\end{algorithm}

The expected value of produced tuples in the {\it Poisson-Olken} algorithm is close to $k$. However, as opposed to reservoir sampling, there is a non-zero probability that {\it Poisson-Olken} may deliver fewer than $k$ tuples. To drastically reduce this chance, one may use a larger value for $k$ in the algorithm and reject the appropriate number of the resulting tuples after the algorithm terminates \cite{Chaudhuri:1999:RSO:304182.304206}. The resulting algorithm will not progressively produce the sampled tuples, but, as our empirical study in Section~\ref{sec:comparison} indicates, it is faster than {\it Reservoir} over large databases with relatively many candidate networks as it does not perform any full join.

%% file: 7.Comparison.tex
\section{Empirical Study}
\label{sec:comparison}

\subsection{Effectiveness}
\subsubsection{Experimental Setup}
It is difficult to evaluate the effectiveness of online and reinforcement learning algorithms for information systems in a live setting with real users because it requires a very long time and a large amount of resources \cite{vorobev2015gathering,hofmann2013balancing,slivkins2013ranked,radlinski2008learning,Grotov:2016:OLR:2911451.2914798}. Thus, most studies in this area use purely simulated user interactions \cite{slivkins2013ranked,radlinski2008learning,hofmann2013balancing}. A notable expectation is \cite{vorobev2015gathering}, which uses a real-world interaction log to simulate a live interaction setting.  We follow a similar approach and use Yahoo! interaction log \cite{yahoo} to simulate interactions using real-world queries and dataset.

{\bf User Strategy Initialization:} 
We train a user strategy over the Yahoo! 43H-interaction log whose details are in Section~\ref{sec:results} using Roth and Erev's method, which is deemed the most accurate to model user learning according to the results of Section~\ref{sec:results}. This strategy has 341 queries and 151 intents. The Yahoo! interaction log contains user clicks on the returned intents, i.e. URLs. However, a user may click a URL by mistake \cite{vorobev2015gathering}. We consider only the clicks that are not noisy according to the relevance judgment information that accompanies the interaction log. According to the empirical study reported in Section~\ref{sec:empAnal}, the parameters of number and length of sessions and the amount of time between consecutive sessions do {\it not} impact the user learning mechanism in long-term communications. Thus, we have {\it not} organized the generated interactions into sessions.

{\bf Metric:} 
Since almost all returned results have only one relevant answer and the relevant answers to all queries have the same level of relevance, we measure the effectiveness of the algorithms using the standard metric of Reciprocal Rank (RR) \cite{IRBook}. RR is $\frac{1}{r}$ where $r$ is the position of the first relevant answer to the query in the list of the returned answers. RR is particularly useful where each query in the workload has a very few relevant answers in the returned results, which is the case for the queries used in our experiment. 

{\bf Algorithms:} 
We compare the algorithm introduced in Section~\ref{sec:arbitrary} against the state-of-the-art and popular algorithm for online learning in information retrieval called UCB-1~\cite{auer2002finite,vorobev2015gathering, radlinski2008learning, moon2012online}. It has been shown to outperform its competitors in several studies \cite{moon2012online, radlinski2008learning}. It calculates a score for an intent $e$ given the $t$th submission of query $q$ as: $Score_t(q, e) =$ $\frac{W_{q,e,t}}{X_{q,e,t}} + \alpha \sqrt{\frac{2ln~t}{X_{q,e,t}}}$, in which $X$ is how many times an intent was shown to the user, $W$ is how many times the user selects a returned intent, and $\alpha$ is the exploration rate set between~$[0,1]$. The first term in the formula prefers the intents that have received relatively more positive feedback, i.e., exploitation, and the second term gives higher scores to the intents that have been shown to the user less often and/or have {\it not} been tried for a relatively long time, i.e., exploration. UCB-1 assumes that users follow a fixed probabilistic strategy. Thus, its goal is to find the fixed but unknown expectation of the relevance of an intent to the input query, which is roughly the first term in the formula; by minimizing the number of unsuccessful trials. 

{\bf Parameter Estimation:}
We randomly select 50\% of the intents in the trained user strategy to learn the exploration parameter $\alpha$ in UCB-1 using grid search and sum of squared errors over 10,000 interactions that are after the interactions in the 43H-interaction log. We do {\it not} use these intents to compare algorithms in our simulation. We calculate the prior probabilities, $\pi$ in Equation~\ref{eqn:payoff}, for the intents in the trained user strategy that are {\it not} used to find the parameter of UCB-1 using the entire Yahoo! interaction log.

{\bf DBMS Strategy Initialization:}
The DBMS starts the interaction with an strategy that does {\it not} have any query. Thus, the DBMS is {\it not} aware of the set of submitted queries apriori. When the DBMS sees a query for the first time, it stores the query in its strategy, assigns equal probabilities for all intents to be returned for this query, returns some intent(s) to answer the query, and stores the user feedback on the returned intent(s) in the DBMS strategy. If the DBMS has already encountered the query, it leverages the previous user's feedback on the results of this query and returns the set of intents for this query using our proposed learning algorithm.
Retrieval systems that leverage online learning perform some filtering over the initial set of answers to make efficient and effective exploration possible \cite{vorobev2015gathering,hofmann2013balancing}. More precisely, to reduce the set of alternatives over a large dataset, online and reinforcement learning algorithms apply a traditional selection algorithm to reduce the number of possible intents to a manageable size. Otherwise, the learning algorithm has to explore and solicit user feedback on numerous items, which takes a very long time. For instance, online learning algorithms used in searching a set of documents, e.g., UCB-1, use traditional information retrieval algorithms to filter out obviously non-relevant answers to the input query, e.g., the documents with low TF-IDF scores. Then, they apply the exploitation-exploration paradigm and solicit user feedback on the remaining candidate answers. The Yahoo! interaction workload has all queries and intents anonymized, thus we are unable to perform a filtering method of our own choosing. Hence, we use the entire collection of possible intents in the portion of the Yahoo! query log used for our simulation. This way, there 4521 intent per query that can be returned, which is close to the number of answers a reinforcement learning algorithm may consider over a large data set after filtering \cite{vorobev2015gathering}. The DBMS strategy for our method is initialized to be completely random.

\subsubsection{Results}
We simulate the interaction of a user population that starts with our trained user strategy with UCB-1 and our algorithm. In each interaction, an intent is randomly picked from the set of intents in the user strategy by its prior probability and submitted to UCB-1 and our method. Afterwards, each algorithm returns a list of 10 answers and the user clicks on the top-ranked answer that is relevant to the query according to the relevance judgment information. We run our simulations for one million interactions. 

Figure~\ref{fig:mrr:100000} shows the accumulated Mean Reciprocal Rank (MRR) over all queries in the simulated interactions. Our method delivers a higher MRR than UCB-1 and its MRR keeps improving over the duration of the interaction. UCB-1, however, increases the MRR at a much slower rate. Since UCB-1 is developed for the case where users do {\it not} change their strategies, it learns and commits to a fixed probabilistic mapping of queries to intents quite early in the interaction. Hence, it cannot learn as effectively as our algorithm where users modify their strategies using a randomized method, such as Roth and Erev's. As our method is more exploratory than UCB-1, it enables users to provide feedback on more varieties of intents than they do for UCB-1. This enables our method to learn more accurately how users express their intents in the long-run. 

We have also observed that our method allows users to try more varieties of queries to express an intent and learn the one(s) that convey the intent effectively. As UCB-1 commits to a certain mapping of a query to an intent early in the interaction, it may {\it not} return sufficiently many relevant answers if the user tries this query to express another intent. This new mapping, however, could be promising in the long-run. Hence, the user and UCB-1 strategies may stabilize in less than desirable states. Since our method does {\it not} commit to a fixed strategy that early, users may try this query for another intent and reinforce the mapping if they get relevant answers. Thus, users have more chances to try and pick a query for an intent that will be learned and mapped effectively to the intent by the DBMS. 

For example, suppose the user has a strategy with two intents and two queries. The user is able to learn from their interactions with the database system. The DBMS strategy also has two intents and two queries and is learning with either UCB-1 or Roth and Erev. The user has been submitting both queries for a single intent for some time. Thus, UCB-1 has learned to return only that single intent when it receives either query. However, when the user does decide to query for its second intent, it is unlikely that UCB-1 will return the second intent. This is due to the fact that UCB-1 commits to a particular strategy quite soon and has difficulty adjusting to a user strategy that changes over time. Roth and Erev, however, changes its strategy at a more gradual rate. Thus, even after some time interacting with the user, it will explore and have a much higher chance at adapting to the changing user strategy.

\begin{figure}[ht!]
    \centering
    \includegraphics[width = 0.9\linewidth]{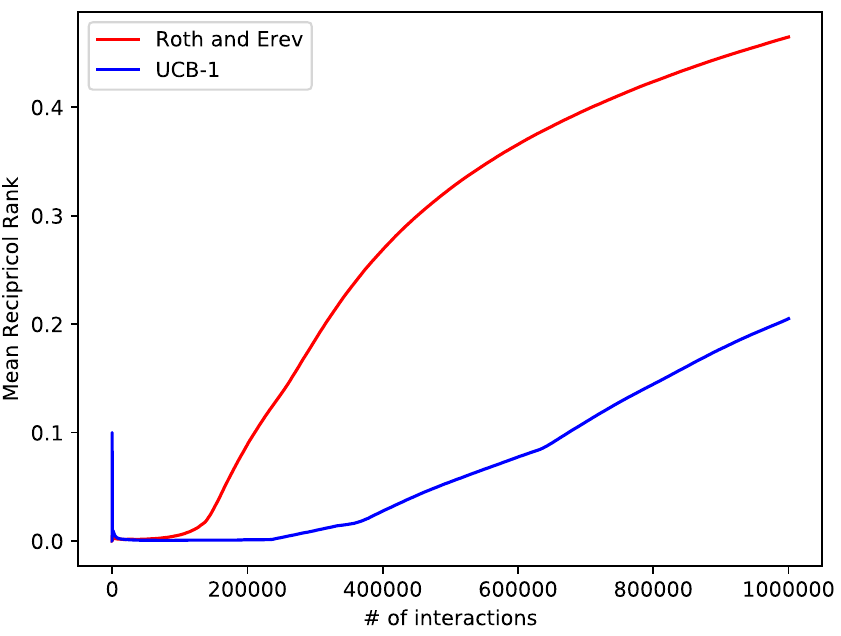}
    \vspace{-10pt}
    \caption{Mean reciprocal rank for 1,000,000 interactions}
    \label{fig:mrr:100000}
\end{figure}

Because our proposed learning algorithm is more exploratory than UCB-1, it may have a longer startup period than UCB-1's. One may pretrain the DBMS strategy using a sample of previously collected user interactions to mitigate this lengthy startup period and improve the effectiveness of answering users' queries in the initial interactions. Such an approach has been used in the context of online learning for document search engines. Another method is for the DBMS to use a less exploratory learning algorithm, such as UCB-1, at the beginning of the interaction. After a certain number of interactions, the DBMS can switch to our proposed learning algorithm. The DBMS can distinguish the time of switching to our algorithm by observing the amount of positive reinforcement it receives from the user. If the user does not provide any or very small number of positive feedback on the returned results, the DBMS is not yet ready to switch to a relatively more exploratory algorithm. If the DBMS observes a relatively large number of positive feedback on sufficiently many queries, it has already provided a relatively accurate answers to many queries. Hence, users may be willing to work with a more exploratory DBMS learning algorithm to discover more relevant answers to their queries. Finally, one may use a relatively large value of reinforcement in the database learning algorithm at the beginning of the interaction to reduce its degree of exploration. The DBMS may switch to a relatively small value of reinforcement after it observes positive feedback on sufficiently many queries.

We have implemented the latter of these methods by increasing the value of reinforcement by some factor. Figure~\ref{fig:mrr:startup} shows the results of applying this technique in our proposed DBMS learning algorithm over the Yahoo! query workload. The value of reinforcement is initially 3 and 6 times larger than the default value proposed in Section~\ref{sec:learning} until a threshold satisfaction value is reached, at which point the reinforcement values scales back down to its original rate.

\begin{figure}[ht!]
    \centering
    \includegraphics[width = 0.9\linewidth]{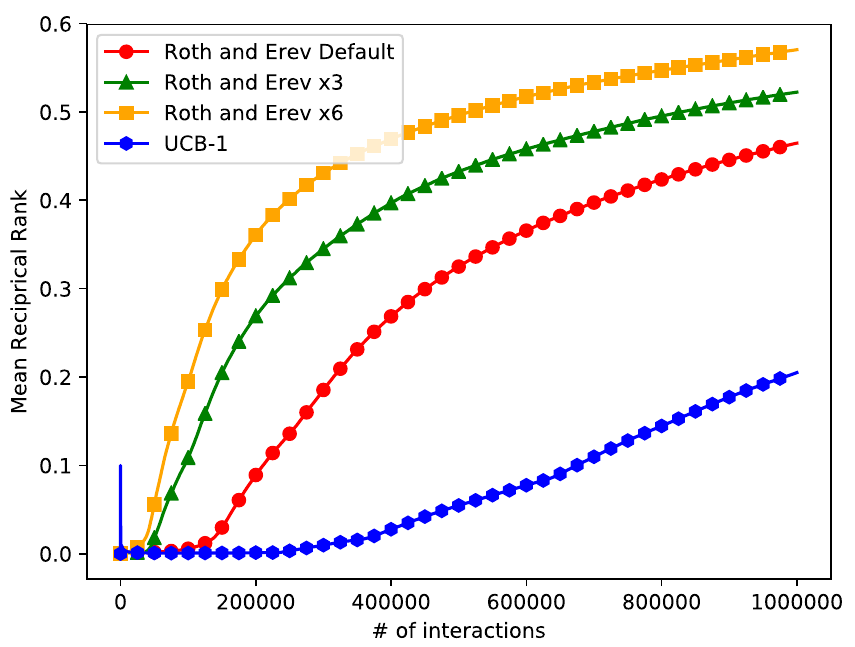}
    \caption{Mean reciprocal rank for 1,000,000 interactions with different degrees of reinforcements}
    \label{fig:mrr:startup}
\end{figure}

We notice that by increasing the reinforcement value by some factor, the startup period is reduced. However, there are some drawbacks to this method. Although we don't see it here, by increasing the rate of reinforcement in the beginning, some amount of exploration may be sacrificed. Thus more exploitation will occur in the beginning of the series of interactions. This may lead to behavior similar to UCB-1 and perform too much exploitation and not enough exploration. Finding the correct degree of reinforcement is an interesting area for future work.

\subsection{Efficiency}
\subsubsection{Experimental Setup}

\textbf{Databases and Queries:}
We have built two databases from Freebase ({\small \it developers.google.com/freebase}), {\it TV-Program} and {\it Play}. {\it TV-Program} contains~7 tables and consisting of~291,026 tuples. {\it Play} contains~3 tables and consisting of 8,685 tuples. 
For our queries, we have used two samples of 621 (459 unique) and 221 (141 unique) queries from Bing ({\it bing.com}) query log whose relevant answers after filtering our noisy clicks, are in {\it TV-program} and {\it Play} databases, respectively \cite{EvaluatingEvidences:Demidova}. After submitting each query and getting some results, we simulate user feedback using the relevance information in the Bing log.  

\textbf{Play:}
Play contains 3 tables, illustrated in table~\ref{table:app:play}. \textit{Fbid} is the identifier for which tuple this is. We use it to determine whether the user was looking for this tuple when they entered a query.
\begin{table}[ht!]
  \centering
  \begin{tabular}{|l|p{3.5cm}|}
    \hline
    Table Name & Attributes\\
    \hline
    \hline
    tbl\_play & id, fbid, name, description\\
    tbl\_genre & id, fbid, name, description\\
    tbl\_play\_genre & id, tbl\_play\_id, tbl\_genre\_id\\
    \hline
  \end{tabular}
  \caption{Play Database Schema}
  \label{table:app:play}
\end{table}

\textbf{TV Program:}
TV Program contains 7 tables, illustrated in table~\ref{table:app:tv}. \textit{Fbid} is is the same as in the Play database.
\begin{table}[ht!]
  \centering
  \begin{tabular}{|l|p{4.2cm}|}
    \hline
    Table Name & Attributes\\
    \hline
    \hline
    tbl\_tv\_program & id, fbid, name, description\\
    tbl\_tv\_program\_tv\_actor & id, tbl\_tv\_program\_id, tbl\_tv\_actor\_id\\
    tbl\_tv\_actor & id, fbid, name, description\\
    tbl\_tv\_program\_tv\_genre & id, tbl\_tv\_program\_id, tbl\_tv\_genre\_id\\
    tbl\_tv\_genre & id, fbid, name, description\\
    tbl\_tv\_program\_tv\_subject & id, tbl\_tv\_program\_id, tbl\_tv\_subject\_id\\
    tbl\_tv\_subject & id, fbid, name, description\\
    \hline
  \end{tabular}
  \caption{Play Database Schema}
  \label{table:app:tv}
\end{table}

\textbf{Query Processing:}
We have used Whoosh inverted index \\({\it whoosh.readthedocs.io}) to index each table in databases. Whoosh recognizes the concept of table with multiple attributes, but cannot perform joins between different tables. Because the {\it Poisson-Olken} algorithm needs indexes over primary and foreign keys used to build candidate network, we have build hash indexes over these tables in Whoosh. Given an index-key, these indexes return the tuple(s) that match these keys inside Whoosh. To provide a fair comparison between {\it Reservoir} and {\it Poisson-Olken}, we have used these indexes to perform join for both methods. We also precompute and maintain all 3-grams of the tuples in each database as mentioned in Section~\ref{sec:system}. We have implemented our system using both {\it Reservoir} and {\it Poisson} algorithms. We have limited the size of each candidate network to 5. Our system returns 10 tuples in each interaction for both methods.

\textbf{Hardware Platform:} We run experiments on a server with 32 2.6GHz Intel Xeon E5-2640 processors with 50GB of main memory. 

\subsubsection{Results}
Table~\ref{table:comparison:avgTimes} depicts the time for processing candidate networks and reporting the results for both {\it Reservoir} and {\it Poisson-Olken} over {\it TV-Program} and {\it Play} databases over 1000 interactions. These results also show that {\it Poisson-Olken} is able to significantly improve the time for executing the joins in the candidate network, shown as performing joins in the table, over {\it Reservoir} in both databases. The improvement is more significant for the larger database, {\it TV-Program}. {\it Poisson-Olken} progressively produces tuples to show to user. But, we are not able to use this feature for all interactions. For a considerable number of interactions, {\it Poisson-Olken} does not produce 10 tuples, as explained in Section~\ref{sec:sampling}. Hence, we have to use a larger value of $k$ and wait for the algorithm to finish in order to find a randomize sample of the answers as explained at the end of Section~\ref{sec:sampling}. Both methods have spent a negligible amount of time to reinforce the features, which indicate that using a rich set of features one can perform and manage reinforcement efficiently.

\begin{table}[ht!]
	\centering
	\small
	\caption{Average candidate networks processing times in seconds for 1000 interactions}
	\begin{tabular}{|l|l|l|l|}
		\hline
		Database & Reservoir & Poisson-Olken \\
		\hline
		\hline
		Play & 0.078 & 0.042\\
		\hline
		TV Program & 0.298 &  0.171 \\
		\hline
	\end{tabular}
	\label{table:comparison:avgTimes}
\end{table}

%% file: 3.Equilibria.tex
\section{Equilibrium Analysis}
\label{sec:equilib}
In this section, we formally investigate the eventual stable states and equilibria of the game. Our analyses hold for every adaptation and 
learning algorithms used by the agents in the game and not only for the methods proposed in this paper.
\subsection{Fixed User Strategy}
In some settings, the strategy of a user may change in a much slower time scale than that of the DBMS. In these cases, it is reasonable to assume that the user's strategy is fixed. Hence, the game will reach a desirable state where the DBMS adapts a strategy that maximizes the expected payoff. Let a {\it strategy profile} be a pair of user and DBMS strategies.  
\begin{definition}
\label{def:bestResponse}
Given a strategy profile ($U$,$D$), $D$ is a best response to $U$ w.r.t. 
effectiveness measure $r$ if we have $u_{r}(U,D)$ $\geq u_{r}(U,D')$ for 
all the database strategies $D'$.
\end{definition}
A DBMS strategy $D$ is a {\it strict best response} to $U$ if the inequality in Definition~\ref{def:bestResponse} becomes strict for all $D'\not= D$.

\begin{example}
\label{example:equilib:strictbest}
Consider the database instance about universities that is shown in Table~\ref{example:equilib:table:instance} and the intents, queries, and the strategy profiles in Tables~\ref{example:equilib:table:intents}, \ref{example:equilib:table:queries}, \ref{example:equilib:table:badstrategy:nash}, and \ref{example:equilib:table:beststrategy:nash}, respectively. 
Given a uniform prior over the intents, the DBMS strategy is a best response the user strategy w.r.t. w.r.t precision and $p@k$ in both strategy profiles \ref{example:equilib:table:badstrategy:nash} and \ref{example:equilib:table:beststrategy:nash}.

\begin{table}[ht!]
        \centering
        \small
        \caption{A database instance of relation Univ}
        \begin{tabular}{l l l l l}
            \hline
            \hline
            Name & Abbreviation & State & Type & Rank\\
            \hline
            Missouri State University & MSU & MO & public & 20\\
            Mississippi State University & MSU & MS & public & 22\\
            Murray State University & MSU & KY & public & 14\\
            Michigan State University & MSU & MI & public & 18\\
            \hline
        \end{tabular}
    \label{example:equilib:table:instance}
\end{table}

\begin{table}[htbp!]
    \caption{Intents and Queries}
    \centering
    \small
    \begin{subtable}{1\linewidth}
        \centering
        \caption{Intents}
        \begin{tabular}{l l}
            \hline
            \hline
            Intent\# & \multicolumn{1}{c}{Intent} \\
            \hline
            $e_1$ & $ans(z) \leftarrow Univ(x,`MSU\textrm', `MS\textrm', y, z)$\\
            $e_2$ & $ans(z) \leftarrow Univ(x,`MSU\textrm', `MI\textrm', y, z)$\\
            $e_3$ & $ans(z) \leftarrow Univ(x,`MSU\textrm', `MO\textrm', y, z)$\\
            \hline
        \end{tabular}
        \label{example:equilib:table:intents}
    \end{subtable}%
    
    \begin{subtable}{1\linewidth}
        \centering
        \caption{Queries}
        \begin{tabular}{l l}
            \hline
            \hline
            Query\# & \multicolumn{1}{c}{Query} \\
            \hline
            $q_1$ & `MSU MI'\\
            $q_2$ & `MSU'\\
            \hline
        \end{tabular}
        \label{example:equilib:table:queries}
    \end{subtable}%
    \label{example:equilib:table:intentsandqueries}
\end{table}
\begin{table}[ht!]
    \caption{Two strategy profiles over the intents and queries in Table~\ref{example:equilib:table:intentsandqueries}. User and DBMS strategies at the top and bottom, respectively.}
    \centering
    \begin{subtable}{.45\linewidth}
        \caption{A strategy profile}
        \begin{tabular}{c|c|c|}
            \cline{2-3}
             & \multicolumn{1}{c|}{$q_1$} &$q_2$\\
            \hline
            \multicolumn{1}{|c|}{$e_1$} & 0 & 1\\
            \hline
            \multicolumn{1}{|c|}{$e_2$} & 0 & 1\\
            \hline
            \multicolumn{1}{|c|}{$e_3$} & 0 & 1\\
            \hline
        \end{tabular}
        \begin{tabular}{c|c|c|c|}
            \cline{2-4}
             & \multicolumn{1}{c|}{$e_1$} & $e_2$ & $e_3$ \\
            \hline
            \multicolumn{1}{|c|}{$q_1$} & 0 & 1 & 0\\
            \hline
            \multicolumn{1}{|c|}{$q_2$} & 0 & 1 & 0\\
            \hline
        \end{tabular}
        \label{example:equilib:table:badstrategy:nash}
    \end{subtable}  
    \begin{subtable}{.45\linewidth}
        \caption{Another strategy profile}
        \begin{tabular}{c|c|c|}
            \cline{2-3}
             & \multicolumn{1}{c|}{$q_1$} &$q_2$\\
            \hline
            \multicolumn{1}{|c|}{$e_1$} & 0 & 1\\
            \hline
            \multicolumn{1}{|c|}{$e_2$} & 1 & 0\\
            \hline
            \multicolumn{1}{|c|}{$e_3$} & 0 & 1\\
            \hline
        \end{tabular}
        \begin{tabular}{c|c|c|c|}
            \cline{2-4}
             & \multicolumn{1}{c|}{$e_1$} & $e_2$ & $e_3$ \\
            \hline
            \multicolumn{1}{|c|}{$q_1$} & 0 & 1 & 0\\
            \hline
            \multicolumn{1}{|c|}{$q_2$} & 0.5 & 0 & 0.5\\
            \hline
        \end{tabular}
        \label{example:equilib:table:beststrategy:nash}
    \end{subtable}
    \label{example:equilib:table:strategies}
\end{table}

\end{example}

\begin{definition}
\label{def:payoff-query}
    Given a strategy profile $(U,D)$, an intent $e_i$, and a query $q_j$, the payoff of $e_i$ using $q_j$ is
     \[u_r(e_i,q_j) = \sum\limits_{\ell=1}^{o} D_{j,\ell} r(e_i,s_{\ell}).\]
\end{definition}
\begin{definition}
\label{def:intentpool} 
    The pool of intents for query $q_j$ in user strategy $U$ is the set of intents $e_i$ such that $U_{i,j} > 0$. 
\end{definition}
\noindent
We denote the pool of intents of $q_j$ as $PL(q_j)$. Our definition of pool of intent resembles the notion of pool of state in signaling games~\cite{Cho:QuarterlyJournalEconomics:1987,Donaldson:MathBiology:2007}. Each result $s_{\ell}$ such that $D_{j,\ell} > 0$ may be returned in response to query $q_j$. We call the set of these results the {\it reply} to query $q_j$.
\begin{definition}
\label{def:bestreply} 
A {\it best reply} to query $q_j$ w.r.t. effectiveness measure $r$ is a reply that maximizes $\sum_{e_i \in PL(q_j)} \pi_i U_{i,j}$ $u_{r}(e_i, q_j)$.
\end{definition}
\noindent
The following characterizes the best response to a strategy.
\begin{lemma}
\label{lemma:BestResponse}  
Given a strategy profile $(U,D)$, $D$ is a best response to $U$ w.r.t. effectiveness measure $r$ if and only if $D$ maps every query to one of its best replies.
\end{lemma}
\begin{proof}
If each query is assigned to its best reply in $D$, no improvement in the expected payoff is possible, thus $D$ is a best response for $U$. 
    Let $D$ be a best response for $U$ such that some query $q$ is not mapped to its best reply in $D$. Let $r_max$ be a best reply for $q$. 
    We create a DBMS strategy $D' \neq D$ such that    
    all queries $q' \neq q$ in $D'$ have the same reply as they have in $D$ and
    the reply of $q$ is $r_{max}$. Clearly, $D'$ has higher payoff than $D$ for $U$. Thus, $D$ is not a best response. 
\end{proof}
\noindent
The following corollary directly results from 
Lemma~\ref{lemma:BestResponse}. 
\begin{corollary}
\label{theorem:BestResponse:strict}  
Given a strategy profile $(U,D)$, $D$ is a strict best response to $U$ w.r.t. effectiveness measure $r$ if and only if every query has one and only one best reply and $D$ maps each query to its best reply. 
\end{corollary}
Given an intent $e$ over database instance $I$, 
some effectiveness measures, such as precision, take their maximum for other results in addition to $e(I)$. For example, given intent $e$, the precision of every non-empty result $s\subset e(I)$ is  equal to the precision of $e(I)$ for $e$. Hence, there are more than one best reply for an intent w.r.t. precision. Thus, according to Corollary~\ref{theorem:BestResponse:strict}, there is not any strict best response w.r.t. precision.

\subsection{Nash Equilibrium}
\label{sec:nash}
In this section and Section~\ref{sec:strictNash}, we analyze the equilibria of the game where both user and DBMS may modify their strategies. A Nash equilibrium for a game is a strategy profile where the DBMS and user will not do better by unilaterally deviating from their strategies.
\begin{definition}
\label{def:Nash}
    A strategy profile $(U,D)$ is a Nash equilibrium w.r.t. a satisfaction function $r$ if $u_{r}(U,D)$ $\geq u_{r}(U',D)$ for all user strategy $U'$ and $u_{r}(U,D) \geq$ $u_{r}(U,D')$ for all database strategy $D'$.
\end{definition}
\begin{example}
\label{example:nashex}
Consider again the database about universities that is shown in Table~\ref{example:equilib:table:instance} and the intents, queries, and the strategy profiles in Tables~\ref{example:equilib:table:intents}, \ref{example:equilib:table:queries}, \ref{example:equilib:table:badstrategy:nash}, and \ref{example:equilib:table:beststrategy:nash}, respectively. 
Both strategy profiles~\ref{example:equilib:table:badstrategy:nash} 
and \ref{example:equilib:table:beststrategy:nash} are 
Nash equilibria w.r.t precision and $p@k$. User and DBMS cannot unilaterally change their strategies and receive a better payoff.
If one modifies the strategy of the database in strategy profile
\ref{example:equilib:table:beststrategy:nash} and replaces the probability of executing and returning $e_1$ and $e_3$ given query $q_2$ to $\epsilon$ and  $1 - \epsilon$, $0 \leq \epsilon \leq 1$,
the resulting strategy profiles are all Nash equilibria. 
\end{example}

Intuitively, the concept of Nash equilibrium captures the fact that users  may explore different ways of articulating and interpreting 
intents, but they may not be able to {\it look ahead} beyond the payoff of a single interaction when adapting their strategies.
Some users may be willing to lose some payoff in the short-term to gain more payoff in the long run, therefore, 
an interesting direction is to define and analyze less myopic equilibria for the game \cite{ArjitaGhosh}. 

If the interaction between user and DBMS reaches a Nash equilibrium, they user 
do not have a strong incentive to change her strategy. 
As a result the strategy of the DBMS and the expected payoff of the game 
will likely to remain unchanged.  
Hence, in a Nash equilibrium the strategies of user and DBMS are likely 
to be stable. Also, the payoff at a Nash equilibrium reflects a potential eventual payoff for the user and DBMS in their interaction. 
Query $q_j$ is a {\it best query} for intent $e_i$ if $q_j \in$ $\arg\max_{q_k} u_r(e_i,q_k)$.

The following lemma characterizes the Nash equilibrium of the game. 
\begin{lemma}\label{lemma:Nash}  
A strategy profile $(U,D)$ is a Nash equilibrium w.r.t. effectiveness measure $r$ if and only if 
\squishlisttwo
\item for every query $q$, $q$ is a best query for every intent $e \in PL(q)$, and 
\item $D$ is a best response to $U$.
\end{list}
\end{lemma}
\begin{proof}
Assume that $(U,D)$ is a Nash equilibrium. Also, assume $q_j$ is not a best query for $e_i \in PL(q_j)$. Let $q_{j'}$ be a best query for $e_i$. 
We first consider the case where $u_{r}(e_i,q_{j'}) > 0$.
We build strategy $U'$ where $U'_{k,\ell} = U_{k,\ell}$ for all entries $(k,\ell) \neq (i,j)$ and $(k,\ell) \neq (i,j')$, $U'_{i,j}=0$, and $U'_{i,j'}=U_{i,j}$. We have $U' \neq U$ and $u_r(U,D) < u_r(U',D)$. Hence, $(U,D)$ is not a Nash equilibrium. Thus, we have $U_{i,j} = 0$ and the first condition of the theorem holds. Now, consider the case where $u_{r}(e_i,q_{j'}) = 0$. In this case, we will also have $u_{r}(e_i,q_{j}) = 0$, which makes $q_j$ a best query for $e_i$. We prove the necessity of the second condition of the theorem similarly. This concludes the proof for the necessity part of the theorem. Now, assume that both conditions of the theorem hold for strategies $U$ and $D$. We can prove that it is not possible to have strategies $U''$ and $D''$ such that $u_r(U,D) < u_r(U'',D)$ or $u_r(U,D) < u_r(U,D'')$ using a similar method.
\end{proof}
\noindent

\subsection{Strict Nash Equilibrium}
\label{sec:strictNash}
A strict Nash equilibrium is a strategy profile in which the DBMS and user will do worse by unilaterally changing their equilibrium strategy.
\begin{definition}\label{def:StrictNash}
A strategy profile $(U,D)$ is a strict Nash equilibrium w.r.t. effectiveness measure $r$ if we have  $u_{r}(U,D) >$ $u_{r}(U,D')$ for all DBMS strategies $D' \neq D$ and $u_{r}(U,D) >$ $u_{r}(U',D)$ for all user strategies $U' \neq U$.
\end{definition}

\begin{table}[htbp]
	\caption{Queries and Intents}
	\centering
    \begin{subtable}{1\linewidth}
        \centering
        \caption{Intents}
        \begin{tabular}{l l}
            \hline
            \hline
            Intent\# & Intent \\
            \hline
            $e_3$ & $ans(z) \leftarrow Univ(x,`MSU\textrm', `MO\textrm', y, z)$\\
            $e_4$ & $ans(z) \leftarrow Univ(x, `MSU\textrm', y, `public\textrm', z)$\\
            $e_5$ & $ans(z) \leftarrow Univ(x,`MSU\textrm', `KY\textrm', y, z)$\\
            \hline
        \end{tabular}
    \label{example:strictequilib:intents}
    \end{subtable}
    \begin{subtable}{1\linewidth}
            \centering
            \caption{Queries}
            \begin{tabular}{l l}
                \hline
                \hline
                Query\# & Query\\
                \hline
                $q_2$ & `MSU'\\
                $q_3$ & `KY'\\
                \hline
            \end{tabular}
        \label{example:strictequilib:queries}
     \end{subtable}
\end{table}

\begin{table}[htbp]
    	\caption{Strict best strategy profile}
        \centering
        \begin{tabular}{c|c|c|}
            \cline{2-3}
             & \multicolumn{1}{c|}{$q_2$} & $q_3$ \\
            \hline
            \multicolumn{1}{|c|}{$e_3$} & 1 & 0\\
            \hline
            \multicolumn{1}{|c|}{$e_4$} & 1 & 0\\
            \hline
            \multicolumn{1}{|c|}{$e_5$} & 0 & 1\\
            \hline
        \end{tabular}
        \begin{tabular}{c|c|c|c|}
            \cline{2-4}
             & \multicolumn{1}{c|}{$e_3$} & $e_4$ & $e_5$ \\
            \hline
            \multicolumn{1}{|c|}{$q_2$} & 1 & 0 & 0\\
            \hline
            \multicolumn{1}{|c|}{$q_3$} & 0 & 0 & 1\\
            \hline
        \end{tabular}
    \label{example:strictequilib:strategy}
\end{table}

\begin{example}
\label{example:equilib:strict}
Consider the intents, queries, strategy profile, and database instance in Tables~\ref{example:strictequilib:intents},~\ref{example:strictequilib:queries},~\ref{example:strictequilib:strategy}, and~\ref{example:equilib:table:instance}. The strategy profile is a strict Nash equilibrium w.r.t precision. However, the strategy profile in Example~\ref{example:nashex} is {\it not} a strict Nash equilibrium as one may modify the value of $D_{q_2,e_1}$ and $D_{q_2,e_3}$ without changing the payoff of the players. 
\end{example}
\noindent
Next, we investigate the characteristics of strategies in a strict Nash equilibria profile. Recall that a strategy is {\it pure} iff it has only 1 or 0 values.
A user strategy is {\it onto} if there is not any query $q_j$ such that $U_{i,j} =0$ 
for all intend $i$. A DBMS strategy is {\it one-to-one} if it does not map two queries to the same result. In other words, there is not any result $s_{ell}$ such that $D_{j\ell} > 0$ and $D_{j'\ell} > 0$ where $j \neq j' $.
\begin{theorem}
\label{theorem:strictNash:Pure}  
If $(U,D)$ is a strict Nash equilibrium w.r.t. satisfaction function $r$, we have 
\begin{itemize}
    \item $U$ is pure and onto. 
    \item $D$ is pure and one-to-one.
\end{itemize}
\end{theorem}
\begin{proof}
Let us assume that there is an intent $e_i$ and a query $q_j$ such that $0 < U_{i,j}< 1$. Since $U$ is row stochastic, there is a query $q_{j'}$ where $0 < U_{i,j'}$ $< 1$. Let $u_r(U_{i,j},D)$ $=\sum_{\ell=1}^{o} D_{j,\ell}r(e_i,s_{\ell})$. If $u_r(U_{i,j},D)$ $= u_r(U_{i,j'},D)$, we can create a new user strategy $U'$ where $U'_{i,j} =1$ and $U'_{i,j'} =0$ and the values of other entries in $U'$ is the same as $U$. Note that the payoff of ($U$,$D$) and ($U'$,$D$) are equal and hence, ($U$,$D$) is not a strict Nash equilibrium.

If $u_r(U_{i,j},D)\not= u_r(U_{i,j'},D)$, without loss of generality one can assume that $u_r(U_{i,j},D)$ $> u_r(U_{i,j'},D)$. We construct a new user strategy $U''$ whose values for all entries except $(i,j)$ and $(i,j')$ are equal to $U$ and $U''_{i,j}=1,$ $U''_{i,j'}=0$. Because $u_r(U,D) < $ $u_r(U'',D)$, ($U$,$D$) is not a strict Nash equilibrium. Hence, $U$ must be a pure strategy. Similarly, it can be shown that $D$ should be a pure strategy. 
 
If $U$ is not onto, there is a query $q_j$ that is not mapped to any intent in $U$. Hence, one may change the value in row $j$ of $D$ without changing the payoff of $(U,D)$.

Assume that $D$ is not one-to-one. Hence, there are queries $q_i$ and $q_j$ and a result $s_{\ell}$ such that $D_{i,\ell} = $ $D_{j,\ell} = 1$. Because $(U,D)$ is a strict Nash, $U$ is pure and we have either $U_{i,\ell}=1$ or $U_{j,\ell}=1$. Assume that $U_{i,\ell}=1$. We can construct strategy $U'$ that have the same values as $U$ for all entries except for $(i,\ell)$ and $(j,\ell)$ and $U'_{i,\ell}=0$, $U'_{j,\ell}=1$. Since the payoffs of $(U,D)$ and  $(U',D)$ are equal, $(U,D)$ is not a strict Nash equilibrium. 
\end{proof}
\noindent
Theorem~\ref{theorem:strictNash:Pure} extends the Theorem 1 in
\cite{Donaldson:MathBiology:2007} for our model. 
In some settings, the user may knows and use fewer queries than intents, i.e., $m > n$. 
Because the DBMS strategy in a strict Nash equilibrium is one-to-one, 
the DBMS strategy does not map some of the results to any query. Hence, the DBMS will never return some results in a strict Nash equilibrium no matter 
what query is submitted. Interestingly, as Example~\ref{example:equilib:strictbest} suggests some of these results may be the results that perfectly satisfy some user's intents. That is, given intent $e_i$ over database instance $I$, the DBMS may never return $e_i(I)$ in a strict Nash equilibrium. 
Using a proof similar to the 
one of Lemma~\ref{lemma:Nash}, we
have the following properties of strict Nash equilibria of a game.
A strategy profile $(U,D)$ is a strict Nash equilibrium w.r.t. effectiveness measure $r$ if and only if:
\squishlisttwo
    \item Every intent $e$ has a unique best query and the user strategy maps $e$ to its best query, i.e., $e \in PL(q_i)$.
    \item $D$ is the strict best response to $U$.
\end{list}
\subsection{Number of Equilibria}
A natural question is how many (strict) Nash equilibria exist in a game. Theorem~\ref{theorem:strictNash:Pure} guarantees that both user and DBMS strategies in a strict Nash  
equilibrium are pure. Thus, given that the sets of intents and queries are finite, 
there are finitely many strict Nash
equilibria in the game. We note that each set of results 
is always finite.
However, we will show that  
if the sets of intents and queries in a game are finite, 
the game has infinite Nash equilibria. 
\begin{lemma}
\label{theorem:infiniteNash}
If a game has a non-strict Nash equilibrium. Then there is an infinitely many Nash equilibria.
\end{lemma}
\begin{proof}
The result follows from the fact that the payoff function \eqref{eqn:payoff} is a bilinear form of $U$ and $D$, i.e.\ it is a linear of $D$ when $U$ is fixed and a linear function of $U$, when $D$ is fixed. If for $D\not = D'$, $(U,D)$ and $(U,D')$ are Nash-equilibria, then $u_r(U,D)=u_r(U,D')$. Therefore $u_r(U,\alpha D+(1-\alpha)D')=u_r(U,D)$ for any $\alpha\in \R$. In particular, for $\alpha\in [0,1]$, if $D,D'$ are stochastic matrices,  $\alpha D+(1-\alpha)D'$ will be a stochastic matrix and hence, $(U,\alpha D+(1-\alpha)D')$ is a Nash equilibrium as well. Similarly, if $(U',D)$  and $(U,D)$ are Nash equilibria for $U\not=U'$, then $u_r(\alpha U+(1-\alpha)U',D)=u_r(U,D)$ and $(\alpha U+(1-\alpha)U',D)$ is a Nash-equilibrium for any $\alpha\in [0,1]$. 
\end{proof}

\begin{theorem}
\label{theorem:infiniteNash2}
Given a game with finitely many intents and queries, 
if the game has a non-strict Nash equilibrium, 
it has an infinite number of Nash equilibria.
\end{theorem}
\begin{proof}
Every finite game has always a mixed Nash equilibrium \cite{GameTheoryBook}. 
According to Theorem~\ref{theorem:strictNash:Pure}, 
a mixed Nash is not a strict Nash equilibrium. 
Therefore, using Lemma~\ref{theorem:infiniteNash}, the game will
have infinitely many Nash equilibria.
\end{proof}

\subsection{Efficiency}
\label{sec:desirability}
In this section we discuss the efficiency of different equilibria. We refer to the value of the 
utility (payoff) in formula~\eqref{eqn:payoff} at a strategy profile to the \textit{efficiency} of the strategy. 
Therefore, the most efficient strategy profile is naturally the one that maximizes \eqref{eqn:payoff}. 
We refer to an equilibria with maximum efficiency as an \textit{efficient equilibrium}.

Thus far we have discussed two types of equilibria, Nash and strict Nash, that once reached it is unlikely that either player will deviate from its current strategy. In some cases it may be possible to enter a state of equilibria where neither player has any incentive to deviate, but that equilibria may not be an efficient equilibrium. 

The strategy profile in Table~\ref{example:framework:table:beststrategy:nash} provides the highest payoff for the user and DBMS given the intents and queries in Tables~\ref{example:framework:table:intents} and~\ref{example:framework:table:queries} over the database in Table~\ref{example:framework:table:instance}. However, some Nash equilibria may not provide high payoffs. For instance, Table~\ref{example:framework:table:badstrategy:nash} depicts another strategy profile for the set of intents and queries in Tables~\ref{example:framework:table:intents} and~\ref{example:framework:table:queries} over the database in Table~\ref{example:framework:table:instance}. In this strategy profile, the user has little knowledge about the database content and expresses all of her intents using a single query $q_2$, which asks for the ranking of universities whose abbreviations are \textit{MSU}. Given query $q_2$, the DBMS always returns the ranking of Michigan State University. Obviously, the DBMS always returns the non-relevant answers for the intents of finding the rankings of Mississippi State University and Missouri State University. If all intents have equal prior probabilities, this strategy profile is a Nash equilibrium. For example, the user will not get a higher payoff by increasing their knowledge about the database and using query $q_1$ to express intent $e_2$. Clearly, the payoff of this strategy profile is less than the strategy profile in Table~\ref{example:framework:table:beststrategy:nash}. Nevertheless, the user and the DBMS do not have any incentive to leave this undesirable stable state once reached and will likely stay in this state.
\begin{definition}
\label{def:optimality:nash}
A strategy profile ($U$,$D$) is optimal w.r.t.\ an effectiveness measure $r$ if we have $u_{r}(U,D) \geq u(U',D')$ for all DBMS strategies $D'$ and $U'$
\end{definition}
\noindent
Since, the games discussed in this paper are games 
of identical interest, i.e.\ the payoff of the user and the DBMS are the same, therefore, an optimal strategy $(U,D)$ (w.r.t.\ an effectiveness measure $r$) is a Nash equilibrium. 
\begin{lemma}
    A strategy ($U$,$D$) is optimal if and only if it is an efficient equilibrium. 
\end{lemma}
\begin{proof}
    Note that if $(U,D)$ is optimal, then none of the two players (i.e.\ the user and the DBMS) has a unilateral incentive to deviate. Therefore $(U,D)$ is a Nash equilibrium. On the other hand, since the payoff function \eqref{eqn:payoff} is a continuous function of $U$ and $D$ and the domain of row-stochastic matrices is a compact space, therefore a maximizer $(U,D)$ of \eqref{eqn:payoff} exists and by the previous part it is a Nash equilibrium. Note that the efficiency of all strategies are bounded by the efficiency of an optimal strategy and hence, any efficient equilibrium is optimal. 
\end{proof}

Similar to the analysis on efficiency of a Nash equilibria, there are strict Nash equilibria that are less efficient than others. Strict Nash equilibria strategy profiles are unlikely to deviate from the current strategy profile, since any unilateral deviation will result in a lower payoff. From this we can say that strict Nash equilibria are also more stable than Nash equilibria since unilateral deviation will always have a lower payoff.


\begin{table}[htbp]
    \centering
    \caption{Strategy Profile 1}
    \begin{subtable}{.45\linewidth}
        \centering
        \caption{User strategy}%
        \begin{tabular}{c|c|c|}
            \cline{2-3}
             & \multicolumn{1}{c|}{$q_1$} &$q_2$\\
            \hline
            \multicolumn{1}{|c|}{$e_1$} & 0 & 1\\
            \hline
            \multicolumn{1}{|c|}{$e_2$} & 1 & 0\\
            \hline
            \multicolumn{1}{|c|}{$e_3$} & 1 & 0\\
            \hline
        \end{tabular}
    \end{subtable}
    \begin{subtable}{.45\linewidth}
        \centering
        \caption{Database strategy}
        \begin{tabular}{c|c|c|c|}
            \cline{2-4}
             & \multicolumn{1}{c|}{$e_1$} & $e_2$ & $e_3$ \\
            \hline
            \multicolumn{1}{|c|}{$q_1$} & 0 & 0 & 1\\
            \hline
            \multicolumn{1}{|c|}{$q_2$} & 1 & 0 & 0\\
            \hline
        \end{tabular}
    \end{subtable}
    \label{example:equilib:desirability:strictNash:less}
\end{table}

\begin{table}[htbp]
    \centering
    \caption{Strategy Profile 2}
    \begin{subtable}{.45\linewidth}
        \centering
        \caption{User Strategy}
        \begin{tabular}{c|c|c|}
            \cline{2-3}
             & \multicolumn{1}{c|}{$q_1$} &$q_2$\\
            \hline
            \multicolumn{1}{|c|}{$e_1$} & 0 & 1\\
            \hline
            \multicolumn{1}{|c|}{$e_2$} & 0 & 1\\
            \hline
            \multicolumn{1}{|c|}{$e_3$} & 1 & 0\\
            \hline
        \end{tabular}
    \end{subtable}
    \begin{subtable}{.45\linewidth}
        \centering
        \caption{Database Strategy}
        \begin{tabular}{c|c|c|c|}
            \cline{2-4}
             & \multicolumn{1}{c|}{$e_1$} & $e_2$ & $e_3$ \\
            \hline
            \multicolumn{1}{|c|}{$q_1$} & 0 & 0 & 1\\
            \hline
            \multicolumn{1}{|c|}{$q_2$} & 0 & 1 & 0\\
            \hline
        \end{tabular}
    \end{subtable}
    \label{example:equilib:desirability:strictNash:more}
\end{table}

As an example of a strict Nash equilibrium that is not efficient, consider both strategy profiles illustrated in Tables~\ref{example:equilib:desirability:strictNash:less} and~\ref{example:equilib:desirability:strictNash:more}. Note that the intents. queries, and results in this example are different from the ones in the previous examples. For this illustration, we set the rewards to $r(e_1,s_1)=1$, $r(e_2,s_2) = 2$, $r(e_2,s_3) = 0.1$, and $r(e_3,s_3) = 3$ where all other rewards are 0. Using our payoff function in Equation~\ref{eqn:payoff} we can calculate the total payoff for the strategy profile in Table~\ref{example:equilib:desirability:strictNash:less} as $u(U,D) = 4.1$. This strategy profile is a strict Nash since any unilateral deviation by either player will result in a strictly worse payoff. Consider the strategy profile in Table~\ref{example:equilib:desirability:strictNash:more} with payoff $u(U,D) = 5$. This payoff is higher than the payoff the strategy profile in Table~\ref{example:equilib:desirability:strictNash:less} receives. It is also not likely for the strategy profile with less payoff to change either strategy to the ones in the strategy profile with higher payoff as both are strict Nash. 

\subsection{Limiting Equilibria of the Investigated Learning Mechanisms}
An important immediate direction to be explored is to study the limiting behavior and equilibria of the game in which the user and DBMS use the 
adaptation mechanisms in Section~\ref{sec:learning}. 
Such questions are partially addressed in \cite{hu2011reinforcement} for the case that both the 
players in a singling game adapt their strategy synchronously and identically. 
In this case, in the limit, only certain equilibria can emerge. 
In particular, eventually, either an object, intent in our model, will be assigned to 
many signals, queries in our model, (synonyms) or many intents can be assigned to one query (polysemy). 
But two intents will not be assigned to two queries (see Theorem 2.3 in \cite{hu2011reinforcement} for more details).
The authors in  \cite{hu2011reinforcement} follow a traditional 
language game approach as explained in Section~\ref{sec:preliminaries}.
Our immediate future research directions on the study of the adaptation mechanism is to 
study the convergence properties of the proposed reinforcement algorithm in Section \ref{sec:arbitrary} 
and Section \ref{sec:userlearning}. 
In particular, when the user is not adapting, does the strategy of 
the DBMS converge to some limiting strategy such as is the best response to the user's strategy? 
Also, when both the user and the DBMS adapt to each other, what equilibria will emerge? 
We believe that answers to these questions will significantly contribute to the proposed 
learning framework in the database systems and provide a novel theoretical perspective on
the efficiency of a query interface.

%% file: 8.RelatedWork.tex
\section{Related Work}
\label{sec:preliminaries}
\textbf{Query learning:}
Database community has proposed several systems that help the DBMS learn the user's information need by showing examples to the user and collecting her feedback~\cite{Maier:VLDB:2015,Dimitriadou:2014:EAQ:2588555.2610523,bonifati:hal-01187986,Tran:SIGMOD:2009,DBLP:conf/pods/AbouziedAPHS13}. In these systems, a user {\it explicitly teaches} the system by labeling a set of examples potentially in several steps without getting any answer to her information need. Thus, the system is broken into two steps: first it learns the information need of the user by soliciting labels on certain examples from the user and then once the learning has completed, it suggests a query that may express the user's information need. These systems usually leverage active learning methods to learn the user intent by showing the fewest possible examples to the user \cite{Dimitriadou:2014:EAQ:2588555.2610523}. However, ideally one would like to have a query interface in which the DBMS learns about the user's intents while answering her (vague) queries as our system does. As opposed to active learning methods, one should combine and balance exploration and learning with the normal query answering to build such a system. Thus, we focus on interaction systems that combine the capabilities of traditional query answering paradigm and leveraging user feedback. Moreover, current query learning systems assume that users follow a fixed strategy for expressing their intents. Also, we focus on the problems that arise in the long-term interaction that contain more than a single query and intent.

\textbf{Game-theoretic Models in Information Systems:}
Game theoretic approaches have been used in various areas of computer science, such as distributed systems, planning, security, and data mining~\cite{Abraham:PODC:2006,Koller+Pfeffer:IJCAI95,GMRS16IJCAI,Shoham:2008:CSG:1378704.1378721,DBLP:conf/icdt/SchusterS15,Gottlob:2001:RMG:375551.375579,Ma:BigScholar:2014}. Researchers have also leveraged economical models to build query interfaces that return desired results to the users using the fewest possible interactions \cite{Cheng:SIGIR:2015}. In particular, researchers have recently applied game-theoretic approaches to model the actions taken by users and document retrieval systems in a single session \cite{Luo:SIGIR:2014}. They propose a framework to find out whether the user likes to continue exploring the current topic or move to another topic. We, however, explore the development of common representations of intents between the user and DMBS. We also investigate the interactions that may contain various sessions and topics. Moreover, we focus on structured rather than unstructured data. Avestani \etal\ have used signaling games to create a shared lexicon between multiple autonomous systems~\cite{Avesani:WWW:2005}. Our work, however, focuses on modeling users' information needs and development of mutual understanding between users and the DBMS. Moreover, as opposed to the autonomous systems, a DBMS and user may update their information about the interaction in different time scales. We also propose novel strategy adaptation mechanism for the DBMS and efficient algorithms based on this mechanism over relational data.

\textbf{Signaling and Language Games:} 
Our game is special case of signaling games, which model communication between two or more agents and have been widely used in economics, sociology, biology, and linguistics~\cite{Lewis:Convention:69,Cho:QuarterlyJournalEconomics:1987,nowak1999evolution,Donaldson:MathBiology:2007}. Generally speaking, in a signaling game a player observes the current state of the world and informs the other player(s) by sending a signal. The other player interprets the signal and makes a decision and/or performs an action that affect the payoff of both players. A signaling game may not be cooperative in which the interests of players do not coincide~\cite{Cho:QuarterlyJournalEconomics:1987}. Our framework extends a particular category of signaling games called language games~\cite{Trapa:MathBiology:2000,nowak1999evolution,Donaldson:MathBiology:2007} and is closely related to learning in signaling games~\cite{hu2011reinforcement,touri2013language,fox2015dynamics}. These games have been used to model the evolution of a population's language in a shared environment. In a language game, the strategy of each player is a stochastic mapping between a set of signals and a set of states. Each player observes its internal state, picks a signal according to its strategy, and sends the signal to inform other player(s) about its state. If the other player(s) interpret the correct state from the sent signal, the communication is successful and both players will be rewarded. Our framework  differs from language games in several fundamental aspects. First, in a language game every player signals, but only one of our players, i.e., user, sends signals. Second, language games model states as an unstructured set of objects. However, each user's intent in our framework is a set of tuples and different intents may intersect. Third, the signals in language games do not posses any particular meaning and can be assigned to every state. A database query, however, restricts its possible answers. Finally, there is not any work on language games on analyzing the dynamics of reinforcement learning where players learn in different time scales.

The authors in \cite{hu2011reinforcement} have also analyzed the effectiveness of a 2-player signaling game in which both players use Roth and Erev's model for learning.  However, they assume that both players learn at the same time-scale. Our result in this section holds for the case where users and DBMS learn at different time-scales, which may arguably be the dominant case in our setting as generally users may learn in a much slower time-scale compared to the DBMS. 

%% file: 9.Conclusion.tex
\section{Conclusion}
\label{sec:conclusion}
Many users do {\it not} know how to express their information needs. A DBMS may interact with these users and learn their information needs. We showed that users learn and modify how they express their information needs during their interaction with the DBMS and modeled the interaction between the user and the DBMS as a game, where the players would like to establish a common mapping from information needs to queries via learning. As current query interfaces do {\it not} effectively learn the information needs behind queries in such a setting, we proposed a reinforcement learning algorithm for the DBMS that learns the querying strategy of the user effectively. We provided efficient implementations of this learning mechanisms over large databases.